\documentclass[review,12pt,authoryear]{elsarticle}
\usepackage[top=1in, bottom=1in, left=1in, right=1in ]{geometry}

\journal{Econometric Reviews}
\usepackage{graphicx,amssymb,amsthm,amsfonts,amsmath,booktabs,caption,subcaption,float,longtable,multirow,enumitem,xcolor,textcomp}
\usepackage[hyphens]{url}
\usepackage{hyperref}
\hypersetup{pdfauthor={Olivia Kvist and J. Eduardo Vera-Vald\'es},pdftitle={Cointegration by Parts: Locating Cointegration in Time}}

\newdefinition{definition}{Definition}
\newtheorem{thm}{Theorem}
\newtheorem{assumption}{Assumption}
\newtheorem{proposition}{Proposition}

\newtheorem{lemma}{Lemma}

\newtheorem{remark}{Remark}

\renewcommand{\thethm}{\arabic{section}.\arabic{thm}}

\renewcommand{\theproposition}{\arabic{section}.\arabic{proposition}}

\renewcommand{\thelemma}{\arabic{section}.\arabic{lemma}}

\renewcommand{\theremark}{\arabic{section}.\arabic{remark}}

\begin{document}

\begin{frontmatter}

\title{Cointegration by Parts: Locating Cointegration in Time}

\author[add1]{Olivia Kvist\corref{cor1}}
%\fnref{fn1}
\ead{omtk@math.aau.dk}

\author[add1,add2]{J. Eduardo Vera-Vald\'es}
\ead{eduardo@math.aau.dk}

\affiliation[add1]{
    organization    = {Department of Mathematical Sciences, Aalborg University},
    postcode        = {9220},
    city            = {Aalborg},
    country         = {Denmark}
}

\affiliation[add2]{
    organization = {Center for Research in Energy: Economics and Markets (CoRE)},
    postcode = {8000},
    city     = {Aarhus},
    country  = {Denmark}
}

\cortext[cor1]{Corresponding author}

\begin{abstract}
  Tests for cointegration are typically applied to a single window spanning the entire sample,
assuming that the long-run relationship holds throughout. When it holds over only a part of the sample, such tests lose power, because the stationary episode is diluted by periods without cointegration. We propose three statistics for testing whether two or more series cointegrate only over a part of the sample, each an infimum of the Engle–Granger statistic over recursive, backward-expanding, or doubly-flexible windows. We derive their limiting distributions and establish which alternatives each is consistent against. Only the doubly-flexible statistic has power against both break directions. Inspired by the seminal work of James G. MacKinnon, critical values are obtained by simulation and summarized through response surface regressions. We apply the tests to global mean sea level and global mean surface temperature anomalies. All three reject the null of no cointegration at the 5\% level, locating it in a sub-period of the 1880--2019 record that coincides with documented discontinuities in sea surface temperature data collection.
\end{abstract}


\end{frontmatter}

\clearpage

\section{Introduction}\label{sec:intro}

Cointegration is a widely used method in econometrics to uncover long-run equilibrium behavior between two or more time series of interest. The concept of multiple non-stationary series having a stationary linear combination was first introduced by \citet{Granger1981SomeSpecification}. Later, the formal Engle-Granger representation theorem was established \citep{Granger1983Co-IntegratedModels, Engle1987Co-IntegrationTesting}. Since its conceptualization, cointegration has been studied quite extensively in the literature. The interest in cointegration is not surprising, since many non-stationary time series, especially in economics, adhere to some long-run equilibrium.

Historically, cointegration has been used in macroeconomics and finance, for example, to test purchasing power parity, the relationship between stock prices and dividends, or the relationship between interest rates and inflation (see, e.g., \citealp{Campbell1987CointegrationModels, Engle1987Co-IntegrationTesting, Mark1990RealInvestigation}). However, the existence of a long-term relationship between two or more variables is not limited to an economic perspective.

In climate research, cointegration analysis has increasingly been carried out in recent years. \cite{Kaufmann2002CointegrationRelations} used cointegration to model the relation between surface temperature and different sources of radiative forcing. \cite{Schmith2012StatisticalMethods} tests for a cointegration relation between global sea level and global mean surface temperature. They find that temperature causally depends on sea level. \cite{Pretis2020EconometricAutoregressions} shows that energy-balance models of the climate are equivalent to a cointegrated system. \cite{Bennedsen2024AFraction} propose a regression-based approach to the estimation of the CO$_2$ airborne fraction. 

At the same time, the maintained assumption in these studies is that the long-run equilibrium relationship is stable across the whole sample, which may not hold in a climate system that is increasingly believed to contain ``tipping elements'' susceptible to abrupt and potentially irreversible regime shifts (see, e.g., \citealp{Lenton2008TippingSystem, Lenton2019ClimateAgainst, ArmstrongMcKay2022ExceedingPoints}). If the relationship between two climate variables holds over only part of the record, a test applied to the full sample can fail to detect it. The appropriate response is then a procedure that searches subsamples rather than one that assumes time invariance.

In this paper, we suggest a family of tests designed for the situation in which a cointegrating relationship holds over part of the sample rather than throughout it. In the classical Engle-Granger [EG] test, the null of no cointegration is tested by examining whether the residuals from the Ordinary Least Squares [OLS] regression of one variable on a group of others follow a unit-root process. The Dickey-Fuller [DF] and Augmented Dickey-Fuller [ADF] tests are commonly used for this purpose (\citealp{Engle1987Co-IntegrationTesting, Gregory1996Residual-basedShifts, MacKinnon2010CriticalTests}). The idea behind our tests is to apply that residual-based test not once, to the full sample, but repeatedly across subsamples, and to retain the window that provides the strongest evidence against a unit root.

It is worth being explicit at the outset about what this provides. Our null hypothesis is that the series do not cointegrate anywhere in the sample, the alternative is that they cointegrate over some sub-interval of it. Therefore, a rejection is evidence that a cointegrating relationship is present somewhere, together with an indication of where the evidence for it is concentrated. It is not, by itself, a test of the hypothesis that a relationship which once held has since broken down, that would require the complementary pairing, with cointegration throughout the sample as the null. We return to this distinction in Section~\ref{sec:conclusion}. What the tests do provide is power in precisely the situation where the standard full-sample test does not, namely when the long-run relationship is confined to a sub-period, and a means of locating that sub-period.

Our test is inspired by research on financial bubbles. In the literature on financial bubbles, the aim is to detect explosive behavior in the price of an asset. 

The most commonly used tests for financial bubbles are the Supremum ADF test [SADF], Backward Supremum ADF test [BSADF], and Generalized Supremum ADF test [GSADF] (\citealp{Phillips2011BehaviorValues, Phillips2014SpecificationBehavior, Phillips2015TestingSP500, Phillips2015TestingDetectors}). We adapt the idea to our context by using left-sided unit root tests, ultimately checking for periods where the cointegrating relationship holds at any point in time.

In most economic applications, there may be no theoretical justification for asking whether it holds only over part of the record. In the climate system, on the contrary, the physical mechanisms linking variables can change over time. Further, the observational record is itself heterogeneous in quality and method, so restricting attention to a single full-sample window is a substantive restriction rather than a neutral default (\citealp{EuropeanCommission2025CausesChange, NationalAeronauticsandSpaceAdministration2025TheChange, U.S.EnvironmentalProtectionAgency2025CausesChange, IPCC2021Summaryeds..}). 

We illustrate our approach on the relationship between global mean sea level and global mean surface temperature over 1880--2019. Under the deterministic specification appropriate for the two drifting series, all three statistics reject the null of no cointegration at the $5\%$ level. The doubly-flexible statistic locates the evidence in 1947--1967, a window that covers the two decades immediately following the wartime change in sea surface temperature measurement documented by \citet{Thompson2008ATemperature}. Hence, what the procedure establishes is that the evidence in this well-studied relationship is not spread evenly across the record.

The paper is structured as follows. In Section \ref{sec:method}, we introduce the methodology behind cointegration and motivate the subsample search. In Section \ref{sec:theory}, we present the three test statistics, derive their limiting distributions, and establish which alternatives each is consistent against. Section \ref{sec:simulation}, contains the Monte Carlo simulation of the critical values and obtains the response surface functions inspired by \citet{MacKinnon1996NumericalTests,MacKinnon2010CriticalTests}. Finally, we make an empirical application of our test in Section \ref{sec:application}.

\section{Methodology}\label{sec:method}

The notion of cointegration testing was first outlined in \cite{Engle1987Co-IntegrationTesting}. They proposed an explicit way of defining whether two or more time series cointegrate, based on specific testing procedures and error-correcting models. In the following, we outline the test for cointegration between two variables to ease notation. Nevertheless, the method can be easily extended to multiple variables.

Let $x_t, y_t \sim I(d)$ be two time series, where $d \in \mathbb{N}$, which means that both $x_t$ and $y_t$ are integrated of order $d$. If there exists a non-trivial linear combination of the two series that is integrated of a smaller order $d-b$, where $d \geq b > 0$, we say that $x_t$ and $y_t$ are cointegrated. The most common values for cointegration testing are $d=b=1$. Hence, assume $x_t, y_t \sim I(1)$; then we say that $x_t$ and $y_t$ cointegrate, with cointegration vector $\beta (\neq 0)$ if, 
\[
z_t = y_t - \beta x_t \sim I(0).
\]
\cite{Engle1987Co-IntegrationTesting} proceed to discuss how to formally test for the existence of $\beta$. The test they propose is closely related to a unit root testing setup, originally formulated by \cite{Fuller1976IntroductionSeries} and \cite{Dickey1979DistributionRoot, Dickey1981LikelihoodRoot}. If a unit root is rejected in the error term, it would be evidence in favor of a cointegrating relationship. From seven different unit root testing methods analyzed in a simulation study, the authors recommend using the ADF test, given that it shows the best power properties. However, new critical values must be obtained since the cointegrating vector must be estimated from the data.

Given this construction, the intuition behind a test on \textit{Cointegration by Parts} is to detect when the error term from the linear regression does not contain a unit root. The explanation of why the test works is straightforward. If two or more series are cointegrated during a specific period, then the OLS regression residuals will be stationary there, but non-stationary elsewhere. Therefore, an EG test applied to the entire sample may fail to reject the unit root, because stationary periods are diluted by observations from periods with no cointegration. By estimating the residuals over many subsamples and computing the corresponding EG statistic for each one, we search for the part of the sample where the evidence against a unit root is strongest. Since more negative EG statistics provide stronger evidence of stationarity, taking the infimum across subsamples identifies the period with the strongest evidence of cointegration. As we will see in Figure~\ref{fig:simulation_GIEG}, the test statistic becomes more negative as the window approaches the cointegrating period, whereas under the unit root it simply fluctuates around its null distribution. 

This motivates applying the EG test recursively over different subsamples and focusing on the subsample that provides the strongest evidence against a unit root. As earlier mentioned, this idea is closely related to the testing methods developed for financial bubbles (see, e.g., \citealp{Keynes1936TheMoney, Samuelson1957IntertemporalSpeculation, Hahn1966EquilibriumGoods, Flood1980MarketTests, Evans1986A1981-84, Evans1991PitfallsPrices}). Later, a sequence of papers developed three tests for explosive behavior based on the supremum of the ADF test statistic. These tests are the earlier mentioned SADF, BSADF, and GSADF (\citealp{Phillips2014SpecificationBehavior, Phillips2015TestingSP500, Phillips2015TestingDetectors}). Each of these tests combines recursive right-sided ADF regressions with a supremum taken over the windows considered. Since we are interested in detecting the existence of cointegration, we perform regular left-sided EG tests. Correspondingly, we take the infimum instead of the supremum of the test statistics.

\subsection{Relation to the Existing Literature}\label{sec:related}

Our tests sit at the intersection of three pieces of literature.

The first is the literature on changes in persistence. \citet{Kim2000DetectionSeries} proposed ratio-based statistics for testing the null of stationarity throughout the sample against the alternative that the process changes from $I(0)$ to $I(1)$ at some unknown date. \citet{Busetti2004TestsPersistence} extended the framework and studied its behavior in both break directions, and \citet{Leybourne2003TestsDifferenceStationarity} developed complementary tests taking difference-stationarity as the null. \citet{Harvey2006ModifiedPersistence} modified these statistics so that they are correctly sized when the process is $I(1)$ throughout, and \citet{Leybourne2007CUSUMPersistence} proposed CUSUM-of-squares versions. The residual process in our Definition~\ref{def:coint} is exactly the object studied in the literature, and our alternative is a change in persistence. The difference is that the process is not observed, rather it is the residual from a cointegrating regression whose coefficient vector must itself be estimated, and which we re-estimate within every window. That estimation step is what generates the nuisance-parameter problem addressed in Lemma~\ref{lem:pivotal}, and it is also why the critical values in Section~\ref{sec:simulation} are not those of the change-in-persistence literature. Estimating $\beta$ once on the full sample and applying an existing persistence test to the resulting residuals is not valid here, precisely because under our alternative the full-sample estimate of $\beta$ is not consistent for the cointegrating vector that holds on the stationary segment.

The second is the literature on recursive and rolling unit root testing. \citet{Banerjee1992RecursiveEvidence} were among the first to study recursive and rolling minimum Dickey-Fuller statistics and to note that their distributions differ from the fixed-sample case. The financial bubble literature discussed above developed the supremum-based versions and the associated limit theory that we adapt \citep{Phillips2011BehaviorValues, Phillips2015TestingSP500, Phillips2015TestingDetectors}. Our contribution is to extend that machinery over from an observed series to a cointegrating residual, and to a left-sided rather than a right-sided alternative.

The third is the literature on instability in cointegrated systems. \citet{Hansen1992TestsProcesses} and \citet{Quintos1993ParameterRegressions} developed tests for parameter constancy in cointegrating regressions. \citet{Gregory1996Residual-basedShifts} and \citet{Gregory1996TestingRelationships} constructed residual-based cointegration tests that allow the cointegrating vector to shift at an unknown date. \citet{Gregory1996Residual-basedShifts} is the closest precedent to our paper, since it also takes the infimum of an ADF statistic over an unknown break date. The distinction is in what varies across the candidates over which the infimum is taken. \citet{Gregory1996Residual-basedShifts} hold the estimation sample fixed at the full sample and vary the form of the regression, adding shift dummies indexed by a candidate break date. Their alternative is cointegration throughout the sample with a coefficient vector that changes. Our alternative is that the cointegrating relation holds only over a sub-interval. Hence, we hold the form of the regression fixed and vary the estimation window.

We begin by presenting the main theoretical result. The supporting constructions, lemmas, and proofs are collected in \ref{app:theory}. Afterward, we conduct a simulation exercise to establish the validity of our theoretical framework. 

\section{Theoretical Establishment}\label{sec:theory}

In this section, we present the formal theorem of the Cointegration by Parts test. The limiting objects it is stated in terms of, the lemmas it rests on, and its proof are given in~\ref{app:theory}.

First, let us formally define what we mean with \textit{Cointegration by Parts}.

\begin{definition}[Model]\label{def:model}
    Consider $N+1$ distinct time series, $y_t$ and
    $x_t=(x_{1,t},\dots,x_{N,t})^\top$ observed for $t = 1,\dots,T$, all integrated of order one, $x_t,y_t\sim I(1)$. They can be represented by
    \begin{equation}\label{eq:model1}
        y_t = \theta^\top d_t + \beta^\top x_t + \varepsilon_t,
    \end{equation}
    where $d_t\in\mathbb{R}^q$ collects the deterministic regressors with coefficient vector $\theta\in\mathbb{R}^q$, $\beta\in\mathbb{R}^N$ is the potential cointegrating vector, and the integration order of the residual process $\varepsilon_t$ determines whether $x_t$ and $y_t$ are indeed cointegrated.
\end{definition}

Following \citet{MacKinnon1996NumericalTests, MacKinnon2010CriticalTests}, we index the procedure by the deterministic specification, since the limiting distributions, and hence the critical values, differ across the three standard cases:
\begin{equation}\label{eq:detcases}
    \text{Case }\mathrm{n}: \theta^\top d_t \equiv 0; \qquad
    \text{Case }\mathrm{c}: d_t = 1; \qquad
    \text{Case }\mathrm{ct}: d_t = (1,\,t/T)^\top .
\end{equation}
Case $\mathrm{n}$ omits the deterministic component altogether, so that the cointegrating relation is forced through the origin, and it is rarely appropriate in practice, since the levels of $x_t$ and $y_t$ typically have an arbitrary origin. Case $\mathrm{c}$ is the default throughout this paper and is the specification underlying the critical values reported in Section~\ref{sec:simulation}. Case $\mathrm{ct}$ is required when $x_t$ and $y_t$ are integrated with drift, in which case the drift dominates the regression and the trend term is needed to recover the cointegration parameter, see the discussion in Section~\ref{sec:application}. The scaling of the trend by $T$ in Case $\mathrm{ct}$ is immaterial for the residuals, since OLS is invariant to a non-singular linear reparameterization of the regressors, but it is convenient for the asymptotics.

Beyond this convention, the deterministic term plays a substantive role for the window-based statistics introduced below. A window beginning at $\lfloor T r_1\rfloor$ with $r_1>0$ opens after the integrated regressors have already accumulated to a level of order $\sqrt{T}$, which the residual process on that window would otherwise inherit as an unmodeled initial condition. Including $d_t$ removes it, since the fitted residuals are orthogonal to $d_t$ by construction and in Case $\mathrm{c}$ therefore sum to zero over every window, this is also why the ADF regression in Equation~\eqref{eq:residmodel} below carries no deterministic terms of its own.

Definition~\ref{def:model} is maintained throughout, under both the null and the alternative hypothesis. What distinguishes the two is the integration order of $\varepsilon_t$, and in particular whether it changes within the sample.

\begin{definition}[Cointegration by Parts]\label{def:coint}
    Let Definition~\ref{def:model} hold. The residual process is said to exhibit a break in the cointegrating relation at the fraction $\tau_0\in(0,1)$ if
    \begin{equation}\label{eq:epsbehave}
        \varepsilon_t\sim
        \begin{cases}
            I(0), & t < \lfloor T \tau_0 \rfloor, \\
            I(1), & t \geq \lfloor T \tau_0 \rfloor.
        \end{cases}
    \end{equation}
    Hence, $x_t$ and $y_t$ are cointegrated before the break and cease to be cointegrated after the break. Analogously, a break in the opposite direction is defined by
    \begin{equation}\label{eq:epsbehave_reverse}
        \varepsilon_t\sim
        \begin{cases}
            I(1), & t < \lfloor T \tau_0 \rfloor,\\
            I(0), & t \geq \lfloor T \tau_0 \rfloor.
        \end{cases}
    \end{equation}
\end{definition}

Neither $\tau_0$ nor the direction of the change in Definition~\ref{def:coint} is known, so the sub-interval on which $\varepsilon_t$ is stationary has to be searched for rather than assumed. That search needs an estimate of $\varepsilon_t$, which is not observed. Equation~\eqref{eq:model1} can be written as
$$\varepsilon_t = y_t - \theta^\top d_t - \beta^\top x_t,$$ 
and $(\theta,\beta)$ can be estimated by OLS \citep{Engle1987Co-IntegrationTesting}. It is convenient to stack the regressors as $z_t := (d_t^\top, x_t^\top)^\top \in \mathbb{R}^{q+N}$ and the coefficients as $\delta := (\theta^\top,\beta^\top)^\top$.

For a recursive window ending at $t = \lfloor Tr \rfloor$, let
\begin{equation}\label{eq:recursive_beta}
    \hat{\delta}_r = \begin{pmatrix}\hat{\theta}_r \\ \hat{\beta}_r\end{pmatrix} = \left(\sum_{t=1}^{\lfloor Tr \rfloor} z_t z_t^\top \right)^{-1} \sum_{t=1}^{\lfloor Tr \rfloor} z_t y_t, \quad r \in [r_0,1],
\end{equation}
where $r_0 \in (0,1)$ denotes the minimum recursive-window fraction. The corresponding recursive residuals are
\begin{equation}\label{eq:recursive_residual}
    \hat{\varepsilon}_{t,r} = y_t - \hat{\theta}_r^\top d_t - \hat{\beta}_r^\top x_t.
\end{equation}

For each recursive endpoint $r$, the residuals are subsequently tested using an ADF
regression over the recursive window,
\begin{equation}
\label{eq:residmodel}
    \Delta\hat{\varepsilon}_{t,r}
      = \gamma_r \hat{\varepsilon}_{t-1,r}
      + \sum_{i=1}^{p} \psi_{i,r}\, \Delta\hat{\varepsilon}_{t-i,r}
      + \nu_{t,r}.
\end{equation}
Equation \eqref{eq:residmodel} is the usual reparameterization of an autoregression
of order $p+1$: if
$$\hat{\varepsilon}_{t,r} = \sum_{j=1}^{p+1} a_{j,r}\hat{\varepsilon}_{t-j,r} + \nu_{t,r},$$
then $\gamma_r = \phi_r - 1$ and $\psi_{i,r} = -\sum_{j=i+1}^{p+1} a_{j,r}$, where
$\phi_r := \sum_{j=1}^{p+1} a_{j,r}$ denotes the sum of the autoregressive
coefficients. The unit-root hypothesis $\phi_r = 1$ is therefore equivalent to
$H_0 : \gamma_r = 0$. The asymptotic results below are stated for a deterministic lag rule
$p = \bar{p}(T)$, common to all windows, under either of two standard regimes. If
$\varepsilon_t$ admits a finite autoregressive representation of order $p_0+1$, then
$\bar{p} = p_0$ is fixed and Equation~\eqref{eq:residmodel} is correctly specified. If
instead $\varepsilon_t$ is of infinite order, $\bar{p}(T)$ must satisfy
$\bar{p} \to \infty$ and $\bar{p}^{3}/T \to 0$ as $T \to \infty$, which ensures that the
autoregressive approximation renders $\nu_{t,r}$ asymptotically serially uncorrelated and
leaves the limiting distribution of the resulting $t$-statistic unchanged
\citep{Said1984TestingOrder}. In the empirical implementation, the lag order is instead
selected in each window using the Bayesian Information Criterion [BIC], following
\citet{Schwarz1978EstimatingModel}, over $p = 0,\dots,p_{\max}$ with the maximum lag given
by Schwert's rule applied to the length $n$ of the window rather than to the full sample,
$p_{\max} = \lfloor 12(n/100)^{1/4} \rfloor$ \citep{Schwert1989TestsInvestigation}. Since
$n \leq T$, this bounds the selected order by $O(T^{1/4})$ uniformly over the windows
considered, so the upper restriction $p^{3}/T \to 0$ holds for every window. BIC does not, however,
deliver a diverging order, so the implementation targets the finite-order regime rather than
the sieve one. Remark~\ref{rem:bic} in~\ref{app:theory} reconciles the data-dependent selection
used in practice with the deterministic rule maintained in the asymptotic theory: if the selected
order is eventually constant across the windows actually scanned, the statistic process computed
with selected lags coincides path by path with the one computed under the deterministic rule, and
the asymptotic results below apply unchanged.

The resulting $t$-statistic is then combined with an infimum operator over the recursive windows to test the null hypothesis of a unit root, corresponding to $H_0:\gamma_r=0$. For each recursive endpoint $r \in [r_0,1]$, define the corresponding EG statistic as,
\begin{equation}
\label{eq:egstat}
    \textit{EG}(r) := \frac{\hat{\gamma}_r}{\operatorname{se}(\hat{\gamma}_r)},
\end{equation}
where $\hat{\gamma}_r$ is estimated by OLS from Equation~\eqref{eq:residmodel} and
$\operatorname{se}(\hat{\gamma}_r)$ denotes the corresponding OLS standard error.

The forward-infimum Engle-Granger [FIEG] statistic is then defined as,
\begin{equation}\label{eq:infegstat}
    \textit{FIEG}(r_0) := \inf_{r\in[r_0,1]}\textit{EG}(r),
\end{equation}
where $r_0 \in (0,1)$ denotes the minimum recursive-window fraction. The recursive endpoint $r$ therefore denotes the fraction of the sample used in each recursive estimation.

Before stating the theorem, we collect the conditions under which the asymptotic results below are derived. Throughout, let
\begin{equation}\label{eq:Rset}
    \mathcal{R}(r_0) := \{(r_1,r_2) : r_2 \in [r_0,1],\; r_1 \in [0, r_2-r_0]\}
\end{equation}
denote the set of admissible windows; the recursive windows $[0,r]$ with $r\in[r_0,1]$ used by $\textit{FIEG}$ correspond to the pairs $(0,r) \in \mathcal{R}(r_0)$.

\begin{assumption}[Joint FCLT and Identification]\label{ass:fclt}
    \begin{enumerate}[label=(\roman*)]
        \item[]
        \item The innovation processes, $\varepsilon_{x,t}$ and $\varepsilon_{y,t}$, associated with $x_t$ and $y_t$ satisfy the joint Functional Central Limit Theorem [FCLT]
        \begin{equation}\label{eq:fclt}
            \frac{1}{\sqrt{T}} \sum_{t=1}^{\lfloor Tr \rfloor}
            \begin{pmatrix}
                \varepsilon_{x,t} \\
                \varepsilon_{y,t}
            \end{pmatrix}
            \Rightarrow
            \begin{pmatrix}
                B_x(r) \\
                B_y(r)
            \end{pmatrix},
            \quad r \in [0,1],
        \end{equation}
        where $(B_x, B_y)^\top$ is an $(N+1)$-dimensional Brownian motion with long-run covariance matrix
        \begin{equation}\label{eq:Omega}
            \Omega = \lim_{T\to\infty} \frac{1}{T}\, \mathbb{E}\left[ \left(\sum_{t=1}^{T}\begin{pmatrix}\varepsilon_{x,t}\\ \varepsilon_{y,t}\end{pmatrix}\right) \left(\sum_{t=1}^{T}\begin{pmatrix}\varepsilon_{x,t}\\ \varepsilon_{y,t}\end{pmatrix}\right)^{\!\top} \right] = \begin{pmatrix} \Omega_{xx} & \Omega_{xy} \\ \Omega_{yx} & \omega_{yy}\end{pmatrix},
        \end{equation}
        assumed positive definite, where $\Omega_{xx}$ is $N \times N$, $\Omega_{xy} = \Omega_{yx}^\top$ is $N \times 1$, and $\omega_{yy}$ is scalar. In general $\Omega_{xy} \neq 0$, so that $B_x$ and $B_y$ are correlated; Lemma~\ref{lem:pivotal} in~\ref{app:theory} shows that the limiting distributions nevertheless do not depend on $\Omega$.
        \item The lag order is given by a deterministic rule $p = \bar{p}(T)$, common to all windows, satisfying one of the two regimes stated after Equation~\eqref{eq:residmodel}.
    \end{enumerate}
\end{assumption}

Part~(i) is the standard FCLT for the partial sums of the innovations. Part~(ii) fixes
the lag rule under which the limit theory is derived, keeping the data-dependence of the empirical lag selection
out of the asymptotics. Because each of the three statistics is an
infimum over a continuum of windows, the continuous-mapping step requires convergence of the process
$\{\textit{EG}(r_1,r_2)\}$ indexed by $(r_1,r_2)$, and not merely of $\textit{EG}(r_1,r_2)$ at each window.
Convergence of that process is established using only Assumption~\ref{ass:fclt}, by
Lemmas~\ref{lem:gram} and~\ref{lem:uniform} in~\ref{app:theory}.

That limit is a functional of Brownian motions projected off the deterministic regressors over the
window in question, the continuous-time counterpart of partialling $d_t$ out of the regression, and
standardized by their own quadratic-variation rate. \ref{app:theory} constructs these objects and defines
from them the functional $\mathcal{Q}(r_1,r_2)$ of Equation~\eqref{eq:Qfunctional}, the Dickey-Fuller
limit written for a process of unit variance rate, which is well defined at every admissible window and
continuous in $(r_1,r_2)$ on $\mathcal{R}(r_0)$.

The two hypotheses under consideration are,
\begin{align}\label{eq:hypotheses}
    H_0 & : \varepsilon_t \sim I(1) \text{ for all } t = 1,\dots,T, \\
    H_1 & : \varepsilon_t \text{ satisfies Definition~\ref{def:coint} for some } \tau_0 \in (0,1), \text{ with the stationary} \nonumber \\
        & \phantom{{}: {}} \text{segment of length at least } r_0. \nonumber
\end{align}
Under $H_0$ the two series do not cointegrate anywhere in the sample, so that the residual process contains a unit root throughout. Under $H_1$ they cointegrate on a sub-interval of the sample. It is convenient to name that sub-interval. Thus, let
\begin{equation}\label{eq:statseg}
    \mathcal{S} := \begin{cases}
        [0,\tau_0] & \text{under Equation~\eqref{eq:epsbehave}, the \emph{forward} break,} \\
        [\tau_0,1] & \text{under Equation~\eqref{eq:epsbehave_reverse}, the \emph{reverse} break,}
    \end{cases}
\end{equation}
denote the stationary segment, and call a window $(r_1,r_2)\in\mathcal{R}(r_0)$ \emph{clean} if $[r_1,r_2]\subseteq\mathcal{S}$. A window that merely \emph{overlaps} $\mathcal{S}$ without being contained in it retains an $I(1)$ segment of non-vanishing length, and the corresponding $\textit{EG}$ statistic remains $O_p(1)$. The length restriction in Equation~\eqref{eq:hypotheses} requires $\tau_0 \geq r_0$ under Equation~\eqref{eq:epsbehave} and $1-\tau_0 \geq r_0$ under Equation~\eqref{eq:epsbehave_reverse}, so that $\mathcal{S}$ is long enough to hold at least one admissible window. Meaning, a cointegration episode shorter than $r_0$ is contained in no admissible window, and no test in this class has power against it. However, length alone is not enough. What drives a statistic to $-\infty$ is a clean window in its own search set, and the three statistics search different sets.

We now formulate the theorem for the asymptotic behavior of the FIEG test statistic under the null hypothesis of a unit root in the residual process.

\begin{thm}[Forward-Infimum Engle-Granger]\label{thm:FIEG}
    Let $x_t$, $y_t$, and $\varepsilon_t$ satisfy Definition~\ref{def:model}, and let Assumption~\ref{ass:fclt} hold. Then, under the null hypothesis $H_0$ in Equation~\eqref{eq:hypotheses},
    \begin{equation}\label{eq:asympnull}
        \textit{FIEG}(r_0) \xrightarrow{d} \inf_{r\in[r_0,1]} \mathcal{Q}(0,r),
    \end{equation}
    with $\mathcal{Q}$ as defined in Equation~\eqref{eq:Qfunctional} in~\ref{app:theory}.
\end{thm}

\ref{app:theory} proves Theorem~\ref{thm:FIEG} and shows in addition that the statistic diverges to
$-\infty$ whenever some recursive window $[0,r]$ with $r\in[r_0,1]$ lies entirely inside a segment on which
$\varepsilon_t$ is stationary. Lemma~\ref{lem:pivotal} then shows that the limit in
Equation~\eqref{eq:asympnull} does not depend on the long-run covariance matrix $\Omega$, but only on $N$,
$r_0$, and the deterministic specification in Equation~\eqref{eq:detcases}. A single tabulation indexed by
those three arguments therefore serves for any data set, and Section~\ref{sec:responsesurface} reports it
as a function of the sample size.

Just as \citet{Phillips2015TestingDetectors} noted, the procedure in Theorem~\ref{thm:FIEG} is sensitive to the number of breaks. Thus, if the cointegrating episode is confined to an interior sub-interval of the sample, so that the relationship neither begins at the start of the record nor runs to its end, the FIEG test suffers from reduced power and, by Proposition~\ref{prop:consistency}, may be inconsistent. Therefore, we adopt their framework of dealing with multiple breaks by proposing the \textit{Backward-Infimum Engle-Granger} [BIEG] and \textit{Generalized-Infimum Engle-Granger} [GIEG] testing procedures. Before we introduce the test statistics and their asymptotic behavior, we need to introduce some more notation. 

The idea behind both the BIEG and GIEG is similar to the BSADF and GSADF from \citet{Phillips2015TestingSP500}. Thus, the BIEG uses a fixed endpoint, $r_2$, and a varying starting point, $r_1$, so that it performs the recursive $\inf \textit{EG}$ test on a backward-expanding sample sequence, while the GIEG allows both the starting and end points of the recursive EG tests to vary. We denote the endpoint of the EG regression by $r_2 \in [r_0, 1]$ and likewise the starting point $r_1 \in [0, r_2 - r_0]$, and write $\textit{EG}(r_1,r_2)$ for the statistic of Equation~\eqref{eq:egstat} computed on the window $[r_1,r_2]$, so that $\textit{EG}(r) = \textit{EG}(0,r)$; \ref{app:theory} sets the window estimator out in Equation~\eqref{eq:beta_window}. The two test statistics are denoted as 
\begin{align}
    \textit{BIEG}(r_0, r_2) & = \inf_{r_1 \in [0, r_2 - r_0]} \textit{EG}(r_1, r_2), \quad r_2 \in [r_0, 1], \label{eq:BIEGstat} \\
    \textit{GIEG}(r_0) & = \inf_{\substack{r_2 \in [r_0, 1] \\ r_1 \in [0, r_2 - r_0]}} \textit{EG}(r_1, r_2). \label{eq:GIEGstat} 
\end{align}
The formal theorems and corresponding proofs for the BIEG and GIEG statistics are given in~\ref{app:BIEGandGIEG}.

Proposition~\ref{prop:consistency} in~\ref{app:theory} shows that a statistic of the form
$\inf_{(r_1,r_2)\in\mathcal{W}}\textit{EG}(r_1,r_2)$ diverges to $-\infty$ if its search set $\mathcal{W}$
contains a clean window, and is $O_p(1)$ when its windows are instead kept a fixed distance away from
$\mathcal{S}$, as those of a one-sided search set are. Every $\textit{FIEG}$ window begins at $0$ and every
$\textit{BIEG}(r_0,1)$ window ends at $1$. Thus, $\textit{FIEG}$ is consistent against a forward break with
$\tau_0 \geq r_0$ and $\textit{BIEG}(r_0,1)$ against a reverse break with $1-\tau_0 \geq r_0$, each
remaining $O_p(1)$ against the other direction. The $\textit{GIEG}$ search set restricts neither endpoint
and is consistent against both, provided only that $\lvert\mathcal{S}\rvert \geq r_0$. Therefore, the restricted
search sets of $\textit{FIEG}$ and $\textit{BIEG}$ buy a computational saving at the cost of
power against one break direction, and only $\textit{GIEG}$ is consistent against $H_1$, as stated in
Equation~\eqref{eq:hypotheses}. Section~\ref{sec:simulation} returns to this trade-off.

\section{Simulation Exercises}\label{sec:simulation}

To verify that the test statistics, specifically the $\textit{GIEG}$ test, behave as we intended, we simulate two time series with a known break structure. They do not cointegrate for the first and last quarter of the sample, but cointegrate in between. 

\begin{figure}[H]
    \centering
    \begin{subfigure}[b]{.45\linewidth}
        \centering
        \includegraphics[width=\textwidth]{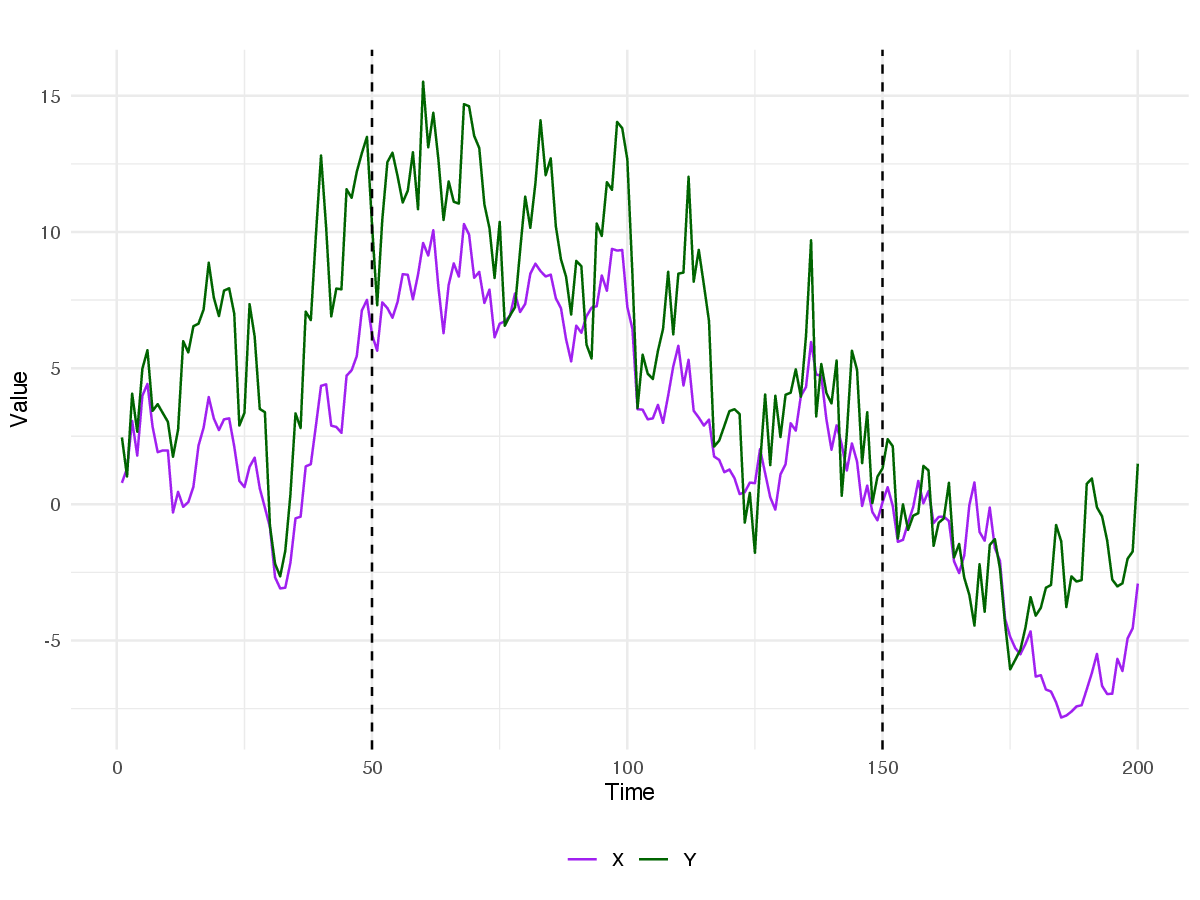}
    \caption{Two simulated time series, $x_t$ and $y_t$.}
    \label{fig:sim_breaks}
    \end{subfigure}
    \begin{subfigure}[b]{.45\linewidth}
        \centering
        \includegraphics[width=\textwidth]{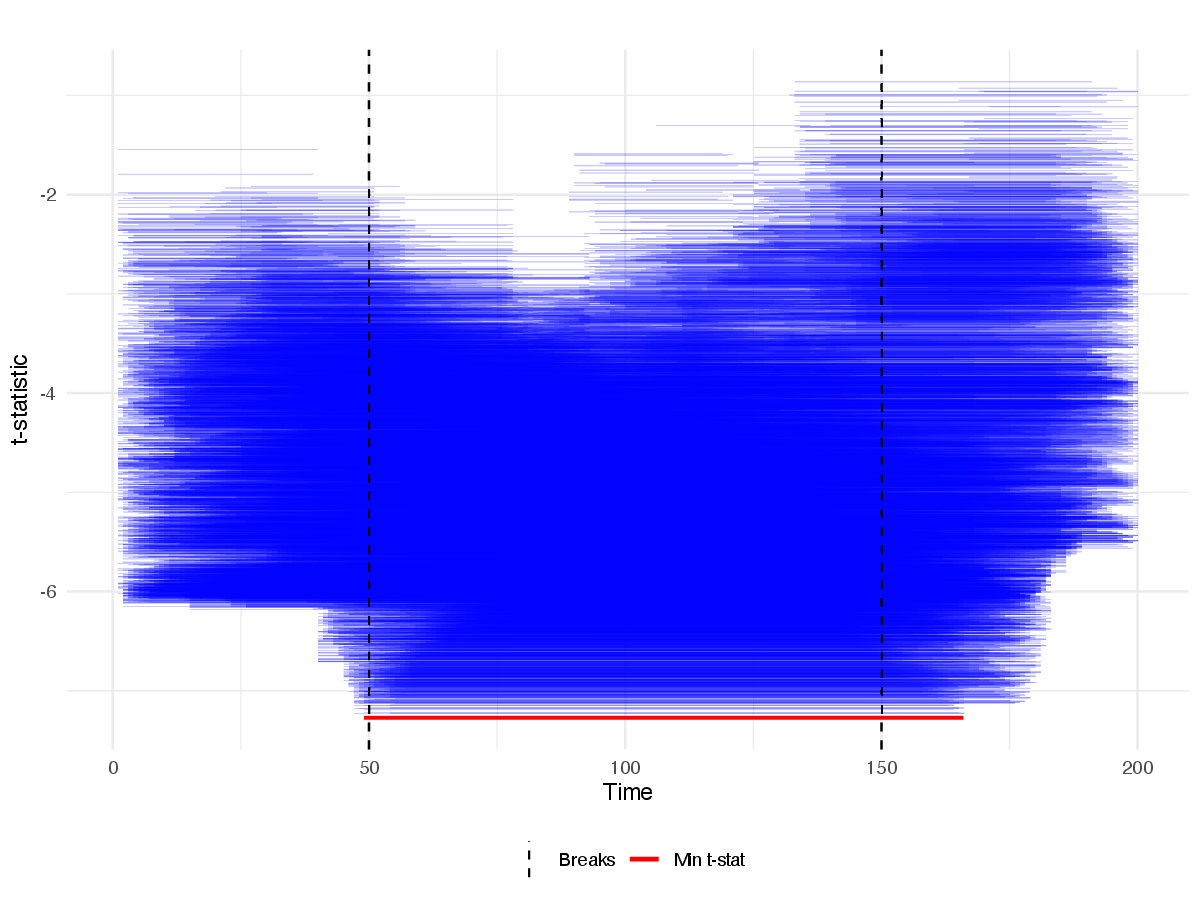}
    \caption{All values of the $t$-statistics.}
    \label{fig:GIEG_sim}
    \end{subfigure}
    \caption{Simulated example of how the $\textit{GIEG}(r_0)$ test works on two time series, $x_t$ and $y_t$.}
    \label{fig:simulation_GIEG}
\end{figure}

Figure~\ref{fig:GIEG_sim} shows the EG statistic for every examined window (light blue), the minimizing window (red), and the true break points (dashed black). The minimizing window falls almost squarely within the cointegrated period, confirming that the test correctly recovers the region of the sample where the relationship actually holds, despite the existence of non-cointegrated periods. 

As with the original ADF and EG tests, our proposed test statistics, also have a non-standard asymptotic distribution, so their critical values must be obtained by simulation. By Lemma~\ref{lem:pivotal} these limiting distributions depend only on $N$, $r_0$, and the deterministic specification, and not on the long-run covariance matrix $\Omega$. This characteristic makes it meaningful to tabulate critical values indexed by $T$ and $N$ alone within a given deterministic case, which parallels the corresponding invariance result for static residual-based cointegration tests in \citet{Phillips1990AsymptoticCointegration}. We tabulate Case $\mathrm{c}$ and Case $\mathrm{ct}$ of Equation~\eqref{eq:detcases} separately, since the two limits differ substantially and, as \citet{MacKinnon1996NumericalTests, MacKinnon2010CriticalTests} emphasizes for the residual-based tests, using the wrong one is a first-order error rather than a refinement. We report the results through one of MacKinnon's seminal contributions, the numerical distribution function: response surfaces fitted at a dense grid of probabilities rather than at three, returning an approximate $p$-value for any statistic at any sample size. Section~\ref{sec:ndf} carries out this estimation, and the results are utilized in Section~\ref{sec:application}.

\subsection{Simulation Design}\label{sec:simdesign}

The statistics of Equations~\eqref{eq:infegstat}, \eqref{eq:BIEGstat}, and \eqref{eq:GIEGstat} are infima over a continuum of windows, so any tabulation of them is a tabulation of one particular finite-sample implementation, and the critical values are usable only if that implementation is stated in full. 

\emph{Trimming.} The minimum window is $\lfloor T r_0 \rfloor$ observations, with $r_0 = 0.15$ throughout, following the convention of the parameter-instability literature \citep{Andrews1993TestsPoint}. We impose $r_0$ as a pure fraction and do not additionally floor the window at a fixed number of observations. Lemma~\ref{lem:pivotal} indexes the limiting distribution by $r_0$, so a window length held fixed in observations would let the effective fraction shrink as $T$ grows, the admissible set $\mathcal{R}(r_0)$ of Equation~\eqref{eq:Rset} would keep expanding, and the simulated quantiles would drift rather than approximate a single limit.

\emph{Lag selection.} Within each window the ADF lag order is chosen by BIC \citep{Schwarz1978EstimatingModel}, as in Remark~\ref{rem:bic}, subject to two restrictions. First, the candidate orders are scored on a common sample, namely the one available at the largest candidate order, so that the criterion compares likelihoods over identical observations. Second, the selected order is the largest $p$ not exceeding Schwert's bound $\lfloor 12(n/100)^{1/4} \rfloor$ \citep{Schwert1989TestsInvestigation}, where $n$ is the window length, that still leaves the ADF regression at least $4(p+1)$ usable observations. Without that cap, Schwert's bound assigns eight lags to a twenty-observation window, leaving two residual degrees of freedom.

\emph{Grid and replications.} We simulate under the null of Definition~\ref{def:model}, with $x_t$ and $y_t$ independent driftless random walks, at $T = 100$, $150$, $200$, $250$, $300$, $400$, $500$, $700$, $900$, and $1{,}000$, and additionally at $T = 1{,}500$ and $T = 2{,}000$ for $N = 1$. We use $10{,}000$ replications for $N = 1$, $5{,}000$ replications for $N = 2, 3$, and $2{,}500$ replication at $T = 1{,}500$ and $T = 2{,}000$. All three statistics are computed from a single scan of each replication, since the \textit{FIEG} and \textit{BIEG} window sets are subsets of the \textit{GIEG} window set. The full replication budget is concentrated at $N \leq 3$; additional tabulations for $N = 4, 5, 6$ on the same grid and with full number of replications are being computed and will be reported in the paper's repository.

Tables~\ref{tab:CVsim_N1} and~\ref{tab:CVsim_N1_ct} report the simulated $1\%$, $5\%$, and $10\%$ quantiles for $N=1$ in the two deterministic cases, each with a bootstrap Monte Carlo standard error, so that movement across $T$ can be separated from simulation noise. 

\begin{table}[H]
    \centering
    \caption{\label{tab:CVsim_N1} Simulated critical values for two time series, that is $N = 1$, for Case $\mathrm{c}$ of Equation~\eqref{eq:detcases}, with minimum window fraction $r_0 = 0.15$.}
    \footnotesize
    \setlength{\tabcolsep}{5pt}
    \begin{tabular}{r rrr rrr rrr}
    \toprule
    & \multicolumn{3}{c}{\textit{FIEG}} & \multicolumn{3}{c}{\textit{BIEG}} & \multicolumn{3}{c}{\textit{GIEG}} \\
    \cmidrule(lr){2-4}\cmidrule(lr){5-7}\cmidrule(lr){8-10}
    $T$ & $1\%$ & $5\%$ & $10\%$ & $1\%$ & $5\%$ & $10\%$ & $1\%$ & $5\%$ & $10\%$ \\
    \midrule
    $100$ & $-5.618$ & $-4.820$ & $-4.420$ & $-5.702$ & $-4.860$ & $-4.456$ & $-8.463$ & $-7.098$ & $-6.546$ \\
     & (0.045) & (0.023) & (0.017) & (0.058) & (0.025) & (0.021) & (0.124) & (0.038) & (0.022) \\[2pt]
    $150$ & $-5.426$ & $-4.613$ & $-4.238$ & $-5.320$ & $-4.582$ & $-4.200$ & $-7.316$ & $-6.371$ & $-5.988$ \\
     & (0.050) & (0.020) & (0.019) & (0.036) & (0.030) & (0.018) & (0.066) & (0.027) & (0.020) \\[2pt]
    $200$ & $-5.116$ & $-4.388$ & $-4.046$ & $-5.135$ & $-4.429$ & $-4.104$ & $-6.723$ & $-6.001$ & $-5.673$ \\
     & (0.047) & (0.020) & (0.015) & (0.044) & (0.022) & (0.019) & (0.045) & (0.020) & (0.015) \\[2pt]
    $250$ & $-4.973$ & $-4.292$ & $-3.984$ & $-5.046$ & $-4.360$ & $-4.047$ & $-6.492$ & $-5.820$ & $-5.516$ \\
     & (0.042) & (0.020) & (0.014) & (0.047) & (0.021) & (0.015) & (0.049) & (0.016) & (0.013) \\[2pt]
    $300$ & $-4.876$ & $-4.207$ & $-3.896$ & $-4.942$ & $-4.320$ & $-4.022$ & $-6.267$ & $-5.643$ & $-5.373$ \\
     & (0.038) & (0.019) & (0.013) & (0.041) & (0.019) & (0.015) & (0.040) & (0.021) & (0.016) \\[2pt]
    $400$ & $-4.738$ & $-4.151$ & $-3.863$ & $-4.924$ & $-4.317$ & $-4.024$ & $-6.116$ & $-5.498$ & $-5.181$ \\
     & (0.034) & (0.018) & (0.012) & (0.039) & (0.017) & (0.013) & (0.029) & (0.016) & (0.012) \\[2pt]
    $500$ & $-4.701$ & $-4.118$ & $-3.860$ & $-4.845$ & $-4.310$ & $-4.059$ & $-5.876$ & $-5.347$ & $-5.068$ \\
     & (0.036) & (0.017) & (0.013) & (0.027) & (0.013) & (0.012) & (0.031) & (0.019) & (0.011) \\[2pt]
    $700$ & $-4.646$ & $-4.103$ & $-3.833$ & $-4.820$ & $-4.316$ & $-4.091$ & $-5.698$ & $-5.233$ & $-4.982$ \\
     & (0.030) & (0.017) & (0.011) & (0.034) & (0.012) & (0.010) & (0.029) & (0.015) & (0.014) \\[2pt]
    $900$ & $-4.631$ & $-4.110$ & $-3.839$ & $-4.916$ & $-4.408$ & $-4.170$ & $-5.639$ & $-5.191$ & $-4.973$ \\
     & (0.036) & (0.018) & (0.013) & (0.037) & (0.014) & (0.012) & (0.027) & (0.014) & (0.010) \\[2pt]
    $1{,}000$ & $-4.613$ & $-4.120$ & $-3.847$ & $-4.921$ & $-4.404$ & $-4.162$ & $-5.684$ & $-5.209$ & $-4.987$ \\
     & (0.037) & (0.018) & (0.013) & (0.035) & (0.013) & (0.012) & (0.022) & (0.013) & (0.010) \\[2pt]
    $1{,}500$ & $-4.636$ & $-4.158$ & $-3.856$ & $-4.974$ & $-4.457$ & $-4.208$ & $-5.644$ & $-5.198$ & $-5.001$ \\
     & (0.064) & (0.033) & (0.026) & (0.043) & (0.032) & (0.023) & (0.043) & (0.018) & (0.020) \\[2pt]
    $2{,}000$ & $-4.597$ & $-4.101$ & $-3.844$ & $-4.905$ & $-4.483$ & $-4.248$ & $-5.690$ & $-5.252$ & $-5.047$ \\
     & (0.073) & (0.030) & (0.021) & (0.037) & (0.033) & (0.018) & (0.048) & (0.026) & (0.018) \\[2pt]
    \bottomrule
    \end{tabular}
    \\[4pt]
    \textbf{Note:} Empirical quantiles of the simulated minima under the null of no cointegration, with bootstrap Monte Carlo standard errors in parentheses ($1{,}000$ bootstrap draws of the quantile). $10{,}000$ replications at $T \leq 1{,}000$, $2{,}500$ at $T = 1{,}500$, and $T = 2{,}000$. Minimum window $\lfloor T r_0 \rfloor$ observations, ADF lag order by BIC on a common sample subject to $n_{\text{eff}} \geq 4(\hat{p}+1)$, as described in Section~\ref{sec:simdesign}.
\end{table}

\begin{table}[H]
    \centering
    \caption{\label{tab:CVsim_N1_ct} Simulated critical values for two time series, that is $N = 1$, for Case $\mathrm{ct}$ of Equation~\eqref{eq:detcases}, with minimum window fraction $r_0 = 0.15$.}
    \footnotesize
    \setlength{\tabcolsep}{5pt}
    \begin{tabular}{r rrr rrr rrr}
    \toprule
    & \multicolumn{3}{c}{\textit{FIEG}} & \multicolumn{3}{c}{\textit{BIEG}} & \multicolumn{3}{c}{\textit{GIEG}} \\
    \cmidrule(lr){2-4}\cmidrule(lr){5-7}\cmidrule(lr){8-10}
    $T$ & $1\%$ & $5\%$ & $10\%$ & $1\%$ & $5\%$ & $10\%$ & $1\%$ & $5\%$ & $10\%$ \\
    \midrule
    $100$ & $-6.314$ & $-5.482$ & $-5.095$ & $-6.540$ & $-5.564$ & $-5.120$ & $-9.669$ & $-8.106$ & $-7.452$ \\
     & (0.039) & (0.026) & (0.018) & (0.057) & (0.028) & (0.021) & (0.089) & (0.043) & (0.025) \\[2pt]
    $150$ & $-5.995$ & $-5.184$ & $-4.814$ & $-6.009$ & $-5.206$ & $-4.839$ & $-8.045$ & $-7.095$ & $-6.677$ \\
     & (0.042) & (0.022) & (0.016) & (0.066) & (0.021) & (0.016) & (0.063) & (0.028) & (0.018) \\[2pt]
    $200$ & $-5.662$ & $-4.996$ & $-4.623$ & $-5.752$ & $-5.004$ & $-4.642$ & $-7.508$ & $-6.648$ & $-6.287$ \\
     & (0.039) & (0.022) & (0.015) & (0.047) & (0.023) & (0.016) & (0.062) & (0.020) & (0.017) \\[2pt]
    $250$ & $-5.500$ & $-4.844$ & $-4.511$ & $-5.531$ & $-4.863$ & $-4.526$ & $-7.094$ & $-6.435$ & $-6.108$ \\
     & (0.034) & (0.019) & (0.017) & (0.042) & (0.019) & (0.017) & (0.039) & (0.014) & (0.016) \\[2pt]
    $300$ & $-5.354$ & $-4.697$ & $-4.398$ & $-5.457$ & $-4.795$ & $-4.472$ & $-6.902$ & $-6.261$ & $-5.957$ \\
     & (0.040) & (0.017) & (0.014) & (0.031) & (0.020) & (0.014) & (0.035) & (0.018) & (0.013) \\[2pt]
    $400$ & $-5.267$ & $-4.633$ & $-4.339$ & $-5.349$ & $-4.722$ & $-4.419$ & $-6.672$ & $-6.025$ & $-5.727$ \\
     & (0.032) & (0.019) & (0.012) & (0.033) & (0.021) & (0.012) & (0.027) & (0.018) & (0.014) \\[2pt]
    $500$ & $-5.118$ & $-4.554$ & $-4.302$ & $-5.266$ & $-4.671$ & $-4.392$ & $-6.437$ & $-5.845$ & $-5.587$ \\
     & (0.031) & (0.021) & (0.011) & (0.030) & (0.016) & (0.014) & (0.033) & (0.016) & (0.012) \\[2pt]
    $700$ & $-5.118$ & $-4.550$ & $-4.289$ & $-5.158$ & $-4.644$ & $-4.379$ & $-6.207$ & $-5.668$ & $-5.405$ \\
     & (0.025) & (0.017) & (0.011) & (0.026) & (0.013) & (0.012) & (0.034) & (0.013) & (0.009) \\[2pt]
    $900$ & $-5.038$ & $-4.535$ & $-4.299$ & $-5.129$ & $-4.654$ & $-4.384$ & $-6.078$ & $-5.587$ & $-5.361$ \\
     & (0.028) & (0.014) & (0.011) & (0.027) & (0.014) & (0.012) & (0.029) & (0.013) & (0.010) \\[2pt]
    $1{,}000$ & $-5.065$ & $-4.524$ & $-4.270$ & $-5.195$ & $-4.688$ & $-4.441$ & $-6.043$ & $-5.576$ & $-5.339$ \\
     & (0.037) & (0.016) & (0.011) & (0.021) & (0.014) & (0.013) & (0.027) & (0.017) & (0.012) \\[2pt]
    $1{,}500$ & $-5.093$ & $-4.547$ & $-4.295$ & $-5.219$ & $-4.719$ & $-4.476$ & $-5.974$ & $-5.578$ & $-5.329$ \\
     & (0.040) & (0.040) & (0.025) & (0.074) & (0.035) & (0.028) & (0.049) & (0.023) & (0.020) \\[2pt]
    $2{,}000$ & $-5.050$ & $-4.555$ & $-4.308$ & $-5.203$ & $-4.773$ & $-4.522$ & $-5.963$ & $-5.532$ & $-5.322$ \\
     & (0.074) & (0.038) & (0.027) & (0.059) & (0.023) & (0.025) & (0.062) & (0.021) & (0.019) \\[2pt]
    \bottomrule
    \end{tabular}
    \\[4pt]
    \textbf{Note:} As Table~\ref{tab:CVsim_N1}, for Case $\mathrm{ct}$: a constant and a linear trend in the cointegrating regression, and no deterministic terms in the residual ADF regression.
\end{table}

For reference, the critical values in Table~\ref{tab:CVsim_N1} are considerably more negative than those conventionally used for a static single-window EG cointegration test. MacKinnon's response surface critical values for a two-variable cointegrating regression with a constant give a $5\%$ critical value of approximately $-3.37$ for $T=200$ and $-3.35$ for $T=500$ \citep{MacKinnon1996NumericalTests, MacKinnon2010CriticalTests}. Since our Case $\mathrm{c}$ uses the same deterministic specification, the comparison is exact: the difference reflects only the infimum taken over windows, and is the reason a subsample search cannot be conducted with off-the-shelf EG critical values.

\subsection{Estimating Response Surface Functions}\label{sec:responsesurface}

MacKinnon’s work has had a major impact on empirical econometrics, and much of it reaches applied researchers through the software they use. Critical values and $p$-values for the ADF and EG tests are reported in the major econometric packages from his response surface simulations \citep{MacKinnon1994ApproximateTests, MacKinnon1996NumericalTests, MacKinnon2010CriticalTests}, and his numerical distribution functions are built into Stata \citep{StataCorp2025StataManual, Schaffer2012EGRANGER:Estimation}, \texttt{R} \citep{Wuertz2025FUnitRoots:Roots, Pfaff2008AnalysisR}, and Python \citep{Seabold2010Statsmodels:Python}, among others. Inspired by his approach, we construct analogous objects for the three statistics of Section~\ref{sec:theory}: the response surfaces in this subsection, and the distribution function built from them in Section~\ref{sec:ndf}. 

To obtain critical values at arbitrary sample sizes, we follow MacKinnon's response surface approach \citep{MacKinnon2010CriticalTests}. Thus, for test $s \in \{ \textit{FIEG}, \textit{BIEG}, \textit{GIEG} \}$, number of regressors $N$, significance level $\alpha$, and sample size $T$, we define
\begin{equation}\label{eq:surfaceresponse}
    \hat{c}_{s,N,\alpha}(T) = {\beta}_{\infty, s, N, \alpha} + \frac{{\beta}_{1,s,N,\alpha}}{T} + \frac{{\beta}_{2,s,N,\alpha}}{T^2} + \frac{{\beta}_{3,s,N,\alpha}}{T^3}.
\end{equation}
In Equation~\eqref{eq:surfaceresponse}, ${\beta}_{\infty}$ is the estimated asymptotic critical value, and ${\beta}_1$, ${\beta}_2$, and ${\beta}_3$ capture finite-sample bias. The three tests reject the null hypothesis of no cointegration when $S_{s,N,T}^{obs} \leq \hat{c}_{s,N,\alpha}(T)$.

We estimate $(\beta_{\infty,s,N,\alpha}, \beta_{1,s,N,\alpha},$ $\beta_{2,s,N,\alpha}, \beta_{3,s,N,\alpha})$ by Weighted Least Squares [WLS]. \citet{MacKinnon1996NumericalTests} noted that OLS estimation of the response surface ignores both the heteroskedasticity of the simulated critical values across $T$ and the correlation across quantiles, and proposed feasible GLS as the remedy. We correct for the heteroskedasticity by weighting each $T$ by the inverse of its estimated variance,
\begin{align*}
    w_T = \widehat{s.e.}\left[\hat{c}_{s,N,\alpha}(T) \right]^{-2},
\end{align*}
so that more precisely estimated points receive proportionally more weight in the fit. Since we fit each quantile in a separate regression rather than jointly, we do not model the cross-quantile correlation component of MacKinnon's GLS.

Because the weights are in units of the Monte Carlo standard error, the weighted residual standard error of the fit,
\begin{align*}
    \hat{\sigma}_{s,N,\alpha}^2 = \frac{1}{n-k}\sum_{T \in \mathcal{T}} w_T \left[ \hat{q}_{s,N,\alpha}(T) - \hat{c}_{s,N,\alpha}(T) \right]^2 ,
\end{align*}
where $\hat{q}_{s,N,\alpha}(T)$ is the simulated quantile, $\mathcal{T}$ is the grid of sample sizes, $n = |\mathcal{T}|$ and $k$ is the number of estimated coefficients, is directly interpretable: $\hat{\sigma} \approx 1$ means the surface reproduces the simulated quantiles to within simulation noise, and $\hat{\sigma} \gg 1$ that the functional form, not the simulation, is the binding constraint. We report it for every fit, together with a standard error for every coefficient. The cubic term is retained only when $|\hat{\beta}_3 / \widehat{s.e.}(\hat{\beta}_3)| \geq 1.96$, and is otherwise dropped and the surface refitted as a quadratic; the table signals this by leaving $\hat{\beta}_3$ blank.

Our design differs from \citet{MacKinnon2010CriticalTests} in how the weights are obtained. MacKinnon splits the simulation at each sample size into many independent replicate sets and treats each resulting quantile estimate as a separate observation. We instead run a single large simulation per sample size and obtain the weight by bootstrap, so each regression has only as many observations as sample sizes in the grid, $12$ for $N=1$ and $10$ for $N=2,3$. Given the cost of repeating the $\textit{GIEG}(r_0)$ scan, this concentrates the replication budget into a more precisely estimated quantile at each point. Since WLS is unbiased under either scheme when the surface is correctly specified, and both weighting schemes reflect the relative precision of each point, the trade-off is one of efficiency and residual degrees of freedom, not of validity.

Table~\ref{tab:response_surface} reports the estimates. Three features deserve comment. First, $\hat{\beta}_{\infty}$ is estimated precisely: its standard error is at most $0.060$ and has a median of $0.031$ across all fits, so the asymptotic critical values are pinned down to roughly the second decimal. Second, $\hat{\sigma}$ is close to unity for \textit{FIEG} and \textit{BIEG} in both cases, with medians between $0.82$ and $1.10$, confirming that Equation~\eqref{eq:surfaceresponse} tracks those two quantile sequences to within simulation error. Third, it is not close to unity for \textit{GIEG}, where the median is $1.47$ in Case $\mathrm{c}$ and the maximum reaches $3.36$. The \textit{GIEG} quantile sequence approaches its limit far more slowly than the other two, the cubic term is retained in only $3$ of the $18$ \textit{GIEG} fits against $23$ of the $36$ \textit{FIEG} and \textit{BIEG} fits, and the consequence is visible in the extrapolation: $\hat{\beta}_{\infty}$ for \textit{GIEG} is less negative than the simulated quantile at the largest available sample size in all $18$ combinations, by between $0.10$ and $0.35$, against Monte Carlo standard errors of at most $0.06$. For \textit{FIEG} and \textit{BIEG} the corresponding discrepancies are of both signs and exceed $0.15$ in absolute value in only three of thirty-six combinations. We therefore regard $\hat{\beta}_{\infty}$ for \textit{GIEG} as an extrapolation that the current grid does not resolve, and we report it as a coefficient of Equation~\eqref{eq:surfaceresponse} rather than as an estimate of the asymptotic critical value. Within the fitted range of $T$, which is the range in which the test is used, this does not affect the accuracy of the fitted values; Section~\ref{sec:ndf} documents that accuracy directly.

\setlength{\tabcolsep}{4pt}
\begin{longtable}[H]{l cc rrrrr}
\caption{Response surface estimates for the simulated critical values
of $FIEG(r_0)$, $BIEG(r_0,1)$, and $GIEG(r_0)$, for $N=1,2,3$
regressors, in both deterministic cases. Coefficients are from Equation~\eqref{eq:surfaceresponse}
estimated by weighted least squares; standard errors in parentheses; $\hat\beta_3$ is blank where the quadratic specification was retained.}\label{tab:response_surface}\tabularnewline
\toprule
Test & $N$ & Level & $\hat\beta_\infty$ & $\hat\beta_1$ & $\hat\beta_2$ & $\hat\beta_3$ & $\hat\sigma$ \\
\endfirsthead
\toprule
Test & $N$ & Level & $\hat\beta_\infty$ & $\hat\beta_1$ & $\hat\beta_2$ & $\hat\beta_3$ & $\hat\sigma$ \\
\midrule
\endhead
\bottomrule
\endlastfoot
\midrule
\multicolumn{8}{l}{\textit{Panel A: Case $\mathrm{c}$}} \\
\midrule
\textit{FIEG} & $1$ & $1\%$ & $-4.622$ & $39.11$ & $-43{,}234.7$ & $2{,}935{,}669$ & $0.32$ \\
 & & & (0.017) & (16.18) & (3{,}953.7) & (256{,}540) & \\
 & $1$ & $5\%$ & $-4.168$ & $93.58$ & $-40{,}067.3$ & $2{,}418{,}344$ & $0.65$ \\
 & & & (0.016) & (15.21) & (3{,}643.5) & (234{,}548) & \\
 & $1$ & $10\%$ & $-3.889$ & $77.40$ & $-31{,}601.4$ & $1{,}854{,}138$ & $0.91$ \\
 & & & (0.016) & (16.02) & (3{,}971.5) & (259{,}592) & \\
 & $2$ & $1\%$ & $-4.886$ & $-131.59$ & $-1{,}806.1$ &  & $1.12$ \\
 & & & (0.050) & (26.98) & (2{,}670.8) &  & \\
 & $2$ & $5\%$ & $-4.427$ & $-85.13$ & $-1{,}815.4$ &  & $1.73$ \\
 & & & (0.034) & (19.43) & (2{,}035.6) &  & \\
 & $2$ & $10\%$ & $-4.314$ & $54.05$ & $-30{,}107.2$ & $1{,}795{,}004$ & $1.06$ \\
 & & & (0.034) & (31.63) & (7{,}423.8) & (471{,}713) & \\
 & $3$ & $1\%$ & $-5.280$ & $-114.07$ & $-5{,}458.9$ &  & $0.82$ \\
 & & & (0.037) & (21.14) & (2{,}239.4) &  & \\
 & $3$ & $5\%$ & $-4.878$ & $5.24$ & $-27{,}195.9$ & $1{,}628{,}214$ & $0.38$ \\
 & & & (0.015) & (14.02) & (3{,}340.6) & (214{,}325) & \\
 & $3$ & $10\%$ & $-4.703$ & $76.14$ & $-38{,}067.7$ & $2{,}259{,}685$ & $0.75$ \\
 & & & (0.025) & (23.29) & (5{,}513.5) & (351{,}794) & \\
\textit{BIEG} & $1$ & $1\%$ & $-4.985$ & $131.51$ & $-42{,}867.8$ & $2{,}264{,}027$ & $1.12$ \\
 & & & (0.044) & (44.29) & (10{,}861.4) & (720{,}431) & \\
 & $1$ & $5\%$ & $-4.570$ & $221.50$ & $-52{,}282.2$ & $2{,}726{,}703$ & $1.38$ \\
 & & & (0.033) & (33.44) & (8{,}486.2) & (561{,}440) & \\
 & $1$ & $10\%$ & $-4.343$ & $224.74$ & $-45{,}666.3$ & $2{,}210{,}900$ & $1.12$ \\
 & & & (0.019) & (19.13) & (4{,}772.7) & (315{,}995) & \\
 & $2$ & $1\%$ & $-5.166$ & $-47.90$ & $-7{,}021.9$ &  & $0.99$ \\
 & & & (0.035) & (18.63) & (1{,}806.0) &  & \\
 & $2$ & $5\%$ & $-4.836$ & $133.63$ & $-41{,}983.0$ & $2{,}303{,}573$ & $0.79$ \\
 & & & (0.031) & (29.29) & (7{,}073.7) & (455{,}906) & \\
 & $2$ & $10\%$ & $-4.576$ & $114.58$ & $-31{,}056.1$ & $1{,}587{,}288$ & $1.10$ \\
 & & & (0.035) & (31.62) & (7{,}391.9) & (468{,}281) & \\
 & $3$ & $1\%$ & $-5.425$ & $-113.09$ & $-4{,}008.9$ &  & $0.48$ \\
 & & & (0.020) & (12.02) & (1{,}250.3) &  & \\
 & $3$ & $5\%$ & $-4.968$ & $-53.95$ & $-3{,}938.5$ &  & $0.81$ \\
 & & & (0.018) & (9.97) & (1{,}038.1) &  & \\
 & $3$ & $10\%$ & $-4.830$ & $64.00$ & $-27{,}143.5$ & $1{,}452{,}916$ & $0.65$ \\
 & & & (0.021) & (19.68) & (4{,}643.5) & (298{,}730) & \\
\textit{GIEG} & $1$ & $1\%$ & $-5.443$ & $-217.52$ & $-8{,}916.9$ &  & $1.83$ \\
 & & & (0.047) & (34.21) & (4{,}488.2) &  & \\
 & $1$ & $5\%$ & $-5.028$ & $-174.92$ & $-3{,}601.0$ &  & $2.71$ \\
 & & & (0.033) & (21.29) & (2{,}495.5) &  & \\
 & $1$ & $10\%$ & $-4.945$ & $-1.67$ & $-42{,}770.4$ & $2{,}708{,}621$ & $3.36$ \\
 & & & (0.054) & (56.12) & (14{,}291.2) & (953{,}053) & \\
 & $2$ & $1\%$ & $-5.870$ & $-257.18$ & $-8{,}133.6$ &  & $0.87$ \\
 & & & (0.037) & (23.06) & (2{,}718.3) &  & \\
 & $2$ & $5\%$ & $-5.396$ & $-223.20$ & $-1{,}671.0$ &  & $0.94$ \\
 & & & (0.021) & (12.20) & (1{,}349.3) &  & \\
 & $2$ & $10\%$ & $-5.172$ & $-206.65$ & $-369.2$ &  & $1.32$ \\
 & & & (0.021) & (11.62) & (1{,}239.9) &  & \\
 & $3$ & $1\%$ & $-6.333$ & $-189.99$ & $-22{,}260.7$ &  & $1.05$ \\
 & & & (0.048) & (28.81) & (3{,}174.5) &  & \\
 & $3$ & $5\%$ & $-5.809$ & $-195.88$ & $-8{,}688.3$ &  & $1.47$ \\
 & & & (0.030) & (18.75) & (2{,}232.9) &  & \\
 & $3$ & $10\%$ & $-5.549$ & $-203.43$ & $-4{,}090.8$ &  & $1.89$ \\
 & & & (0.030) & (16.97) & (1{,}970.9) &  & \\
\midrule
\multicolumn{8}{l}{\textit{Panel B: Case $\mathrm{ct}$}} \\
\midrule
\textit{FIEG} & $1$ & $1\%$ & $-5.043$ & $15.32$ & $-42{,}793.3$ & $2{,}856{,}103$ & $0.95$ \\
 & & & (0.042) & (41.20) & (10{,}060.6) & (651{,}783) & \\
 & $1$ & $5\%$ & $-4.542$ & $47.80$ & $-38{,}713.4$ & $2{,}458{,}907$ & $1.20$ \\
 & & & (0.030) & (29.74) & (7{,}275.2) & (473{,}459) & \\
 & $1$ & $10\%$ & $-4.319$ & $68.30$ & $-36{,}050.1$ & $2{,}147{,}638$ & $1.07$ \\
 & & & (0.019) & (18.56) & (4{,}471.2) & (289{,}152) & \\
 & $2$ & $1\%$ & $-5.436$ & $15.77$ & $-40{,}112.1$ & $2{,}412{,}202$ & $0.74$ \\
 & & & (0.060) & (58.67) & (14{,}200.8) & (917{,}217) & \\
 & $2$ & $5\%$ & $-4.836$ & $-21.33$ & $-23{,}126.6$ & $1{,}379{,}299$ & $1.10$ \\
 & & & (0.045) & (43.06) & (10{,}314.2) & (662{,}801) & \\
 & $2$ & $10\%$ & $-4.591$ & $-22.27$ & $-17{,}159.3$ & $983{,}975$ & $1.00$ \\
 & & & (0.033) & (30.44) & (7{,}228.5) & (463{,}002) & \\
 & $3$ & $1\%$ & $-5.521$ & $-171.87$ & $-2{,}468.2$ &  & $1.11$ \\
 & & & (0.047) & (29.89) & (3{,}659.0) &  & \\
 & $3$ & $5\%$ & $-5.092$ & $-118.50$ & $-828.5$ &  & $1.10$ \\
 & & & (0.024) & (12.81) & (1{,}210.5) &  & \\
 & $3$ & $10\%$ & $-4.866$ & $-96.29$ & $-980.7$ &  & $1.34$ \\
 & & & (0.023) & (12.83) & (1{,}283.5) &  & \\
\textit{BIEG} & $1$ & $1\%$ & $-5.137$ & $-0.21$ & $-33{,}853.3$ & $1{,}989{,}958$ & $1.18$ \\
 & & & (0.053) & (56.75) & (14{,}887.7) & (1{,}001{,}095) & \\
 & $1$ & $5\%$ & $-4.762$ & $126.70$ & $-47{,}937.8$ & $2{,}734{,}736$ & $1.62$ \\
 & & & (0.036) & (36.19) & (8{,}872.6) & (581{,}905) & \\
 & $1$ & $10\%$ & $-4.524$ & $148.47$ & $-47{,}630.3$ & $2{,}686{,}792$ & $1.68$ \\
 & & & (0.032) & (30.50) & (7{,}363.8) & (479{,}100) & \\
 & $2$ & $1\%$ & $-5.292$ & $-157.50$ & $-1{,}889.7$ &  & $0.91$ \\
 & & & (0.042) & (22.78) & (2{,}433.6) &  & \\
 & $2$ & $5\%$ & $-4.974$ & $26.54$ & $-33{,}101.3$ & $2{,}024{,}194$ & $0.91$ \\
 & & & (0.039) & (37.42) & (9{,}100.2) & (589{,}596) & \\
 & $2$ & $10\%$ & $-4.764$ & $60.27$ & $-32{,}321.0$ & $1{,}857{,}014$ & $0.93$ \\
 & & & (0.030) & (27.66) & (6{,}624.7) & (432{,}727) & \\
 & $3$ & $1\%$ & $-5.738$ & $-115.29$ & $-4{,}766.6$ &  & $0.88$ \\
 & & & (0.038) & (22.16) & (2{,}471.8) &  & \\
 & $3$ & $5\%$ & $-5.225$ & $-80.73$ & $-3{,}393.2$ &  & $0.77$ \\
 & & & (0.017) & (9.14) & (876.1) &  & \\
 & $3$ & $10\%$ & $-4.978$ & $-63.89$ & $-3{,}036.4$ &  & $1.43$ \\
 & & & (0.025) & (14.12) & (1{,}455.8) &  & \\
\textit{GIEG} & $1$ & $1\%$ & $-5.635$ & $-456.98$ & $30{,}907.1$ & $-2{,}547{,}899$ & $1.17$ \\
 & & & (0.053) & (55.79) & (14{,}908.8) & (1{,}035{,}417) & \\
 & $1$ & $5\%$ & $-5.322$ & $-267.53$ & $-849.0$ &  & $2.56$ \\
 & & & (0.032) & (20.21) & (2{,}476.1) &  & \\
 & $1$ & $10\%$ & $-5.098$ & $-250.30$ & $1{,}552.8$ &  & $3.07$ \\
 & & & (0.030) & (19.02) & (2{,}101.2) &  & \\
 & $2$ & $1\%$ & $-6.241$ & $-227.20$ & $-18{,}697.5$ &  & $1.07$ \\
 & & & (0.043) & (29.43) & (3{,}491.9) &  & \\
 & $2$ & $5\%$ & $-5.626$ & $-347.95$ & $20{,}974.4$ & $-1{,}808{,}307$ & $0.68$ \\
 & & & (0.029) & (26.72) & (6{,}432.1) & (436{,}249) & \\
 & $2$ & $10\%$ & $-5.494$ & $-241.75$ & $-2{,}093.8$ &  & $1.08$ \\
 & & & (0.018) & (10.58) & (1{,}134.6) &  & \\
 & $3$ & $1\%$ & $-6.555$ & $-209.62$ & $-26{,}860.9$ &  & $1.28$ \\
 & & & (0.053) & (35.50) & (4{,}223.7) &  & \\
 & $3$ & $5\%$ & $-6.117$ & $-207.86$ & $-13{,}768.7$ &  & $1.10$ \\
 & & & (0.024) & (14.19) & (1{,}688.1) &  & \\
 & $3$ & $10\%$ & $-5.876$ & $-210.87$ & $-8{,}244.2$ &  & $0.59$ \\
 & & & (0.010) & (6.08) & (671.6) &  & \\
\end{longtable}
\noindent\footnotesize \textbf{Note:} $\hat\sigma$ is the weighted residual standard error defined in the text; because the weights are inverse squared Monte Carlo standard errors, $\hat\sigma \approx 1$ indicates that the surface fits the simulated quantiles to within simulation noise. The cubic term is retained when $|\hat\beta_3|$ exceeds $1.96$ standard errors and dropped otherwise.
\normalsize

\subsection{A Numerical Distribution Function}\label{sec:ndf}

Response surfaces at three conventional levels tell a user whether a statistic clears a threshold, not how much evidence it carries. Following \citet{MacKinnon1994ApproximateTests, MacKinnon1996NumericalTests}, we therefore construct a numerical distribution function, which converts any observed statistic at any sample size into an approximate $p$-value. This is done by fitting response surfaces at a dense grid of quantiles rather than at three and interpolating between them. As in MacKinnon's own tabulations, the result is distributed rather than printed: $184$ surfaces at each of the $18$ combinations of test, $N$, and deterministic case come to more than $3,000$ in all, which no set of printed tables can reasonably carry. The replication repository provides them, alongside an implementation of the three tests, so that Section~\ref{sec:application} can report $p$-values rather than only rejections.

We fit Equation~\eqref{eq:surfaceresponse}, by the same WLS and with the same rule for the cubic term, at a grid $\mathcal{P}$ of $184$ probabilities, dense in the left tail where a left-sided test lives: $36$ points below $0.01$, $75$ between $0.011$ and $0.20$, and $73$ above $0.21$.

Given an observed statistic $S^{obs}$ computed on a sample of size $T$ with $N$ regressors, the approximate $p$-value is obtained in three steps. First, evaluate the fitted critical value $\hat{c}_{s,N,p}(T)$ at every $p \in \mathcal{P}$ from Equation~\eqref{eq:surfaceresponse}, which gives $184$ points on the fitted distribution function at that sample size. Second, locate $S^{obs}$ between two adjacent points $\hat{c}_{s,N,p_{i-1}}(T) < S^{obs} \leq \hat{c}_{s,N,p_i}(T)$. Third, interpolate linearly not in $p$ but in its normal quantile,
\begin{equation}\label{eq:ndf}
    \hat{p}\left(S^{obs}, T\right) = \Phi\!\left( \Phi^{-1}(p_{i-1}) + \lambda \left[ \Phi^{-1}(p_{i}) - \Phi^{-1}(p_{i-1}) \right] \right), \quad
    \lambda = \frac{S^{obs} - \hat{c}_{s,N,p_{i-1}}(T)}{\hat{c}_{s,N,p_{i}}(T) - \hat{c}_{s,N,p_{i-1}}(T)},
\end{equation}
where $\Phi$ is the standard normal distribution function. Interpolating in $\Phi^{-1}(p)$ rather than in $p$ is the device \citet{MacKinnon1994ApproximateTests} uses to keep the approximation accurate far into the tail. Statistics beyond the ends of the grid are reported as $p < 0.0001$ or $p > 0.99$ rather than extrapolated.

Figure~\ref{fig:ndf_dist} shows the resulting function at $T = 140$, $N = 1$, in Case $\mathrm{ct}$, the configuration used in Section~\ref{sec:application}. The three diamonds mark the statistics reported there. The \textit{FIEG} and \textit{BIEG} curves are nearly indistinguishable, while the \textit{GIEG} curve shifts left by roughly $2.0$ at the $5\%$ level, which is the cost of the doubly-flexible search expressed as a distribution rather than a single number.
Equation~\eqref{eq:ndf} reproduces the quantiles it was fitted to, with mean absolute deviations of $0.0009$, $0.0023$, and $0.0035$ at the $1\%$, $5\%$, and $10\%$ levels and a maximum deviation of $0.028$. The harder test is off the grid. 
Figure~\ref{fig:ndf_cal} plots the whole empirical distribution function of those $p$-values against the $45$-degree line at $T = 140$.
\begin{figure}[H]
    \centering
    \begin{subfigure}[H]{.45\linewidth}
        \centering
        \includegraphics[width = \textwidth]{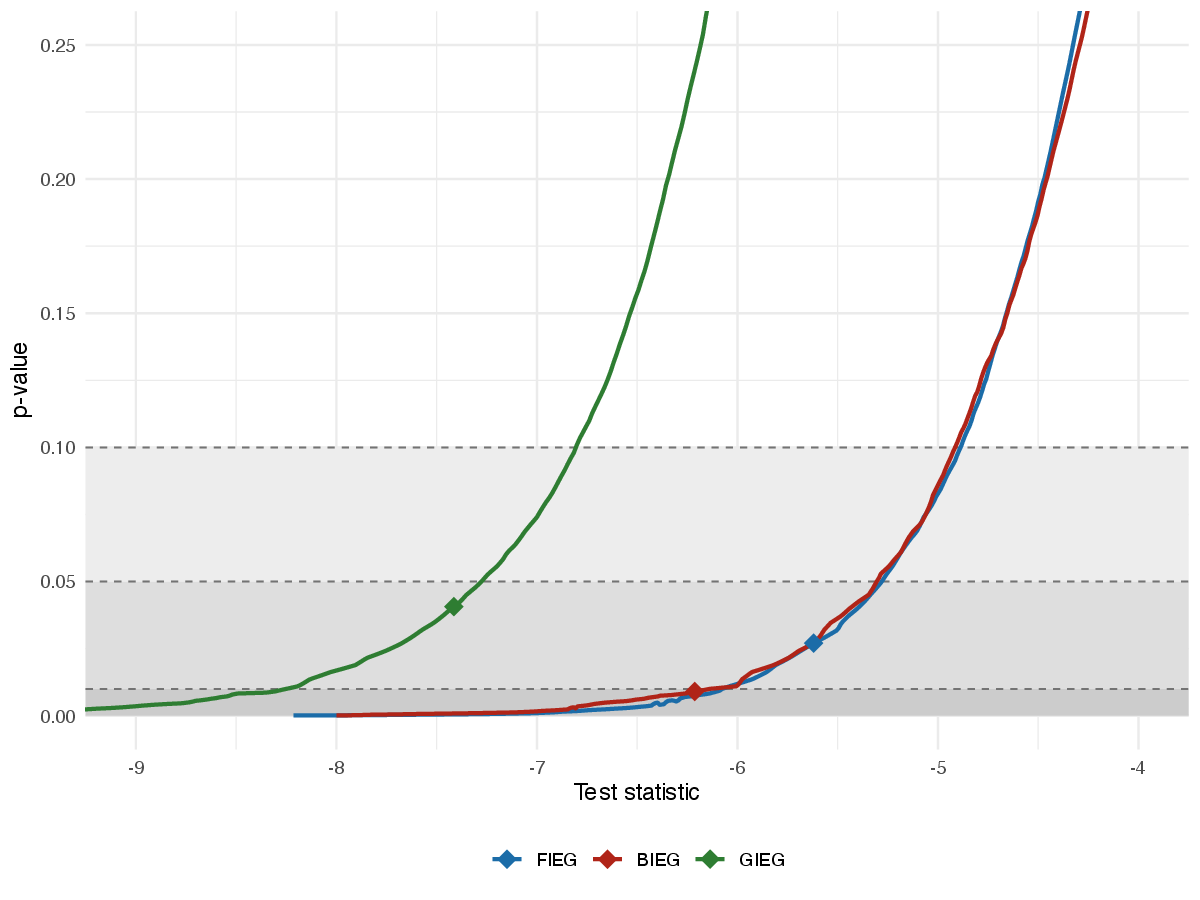}
        \caption{Case $\mathrm{ct}$, $N = 1$, and $T = 140$.}
        \label{fig:ndf_dist}
    \end{subfigure}
    \begin{subfigure}[H]{.45\linewidth}
        \centering
        \includegraphics[width = \textwidth]{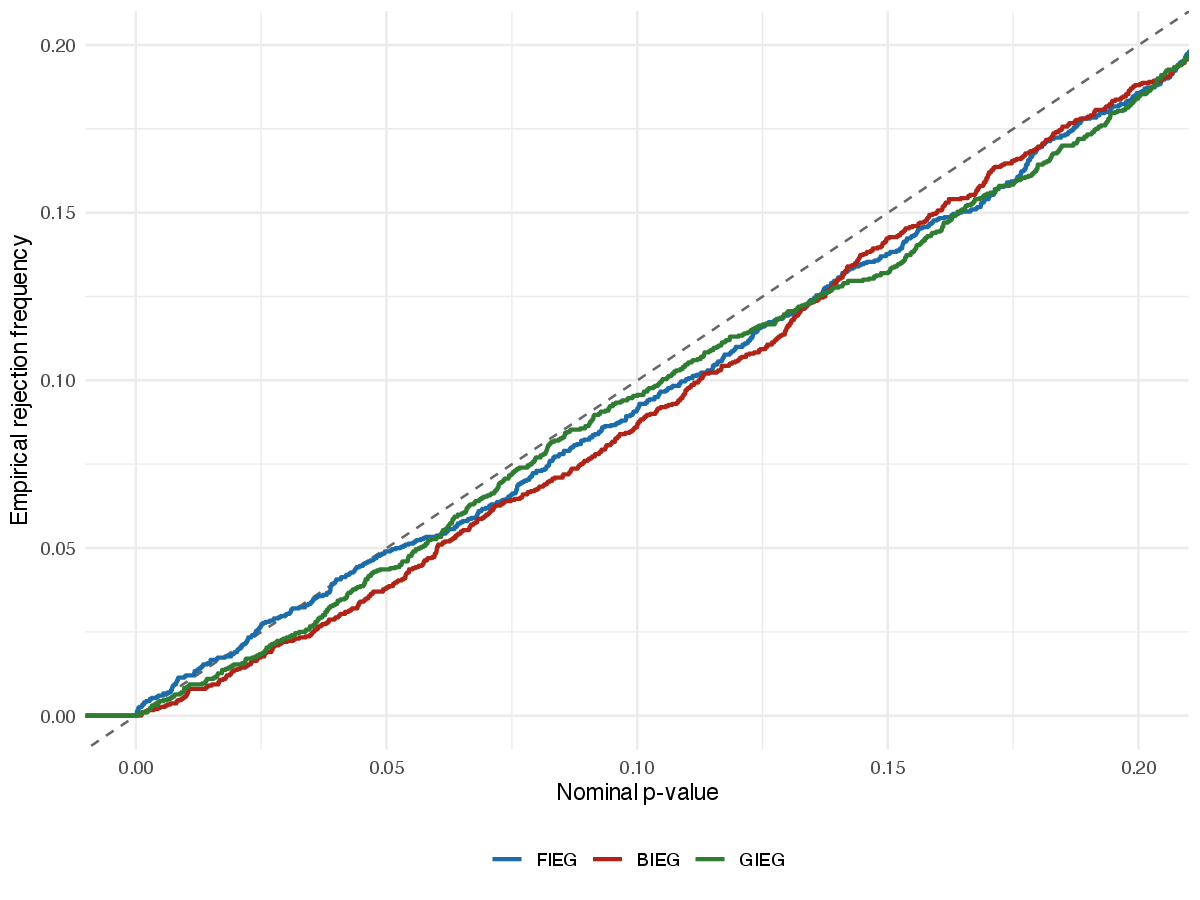}
        \caption{Calibration on $3,000$ independent null draws.}
        \label{fig:ndf_cal}
    \end{subfigure}
    \caption{The numerical distribution function. Panel (a): fitted $p$-value as a function of the observed statistic at $T = 140$, $N = 1$, Case $\mathrm{ct}$, with the dashed lines marking the $1\%$, $5\%$, and $10\%$ levels and the diamonds marking the statistics of Table~\ref{tab:application_results}. Panel (b): empirical distribution function of the $p$-values assigned by Equation~\eqref{eq:ndf} to $3{,}000$ independent null replications at $T = 140$, against the $45$-degree line.}
    \label{fig:ndf}
\end{figure}

Table~\ref{tab:ndf_calibration} reports the fraction of $3{,}000$ independent null replications assigned a $p$-value at or below each nominal level, at four sample sizes, none of which was used in fitting, a correctly calibrated function returns fractions equal to the nominal levels. 

\footnotesize
\setlength{\tabcolsep}{5pt}
\begin{longtable}{rr rrr rrr rrr}
\caption{\label{tab:ndf_calibration} Calibration of the numerical distribution function: fraction of $3{,}000$ independent null replications assigned a $p$-value at or below each nominal level. None of the four sample sizes belongs to the grid on which the response surfaces were fitted.} \\
\toprule
& & \multicolumn{3}{c}{\textit{FIEG}} & \multicolumn{3}{c}{\textit{BIEG}} & \multicolumn{3}{c}{\textit{GIEG}} \\
\cmidrule(lr){3-5}\cmidrule(lr){6-8}\cmidrule(lr){9-11}
$T$ & $N$ & $0.01$ & $0.05$ & $0.10$ & $0.01$ & $0.05$ & $0.10$ & $0.01$ & $0.05$ & $0.10$ \\
\midrule
\endfirsthead
\multicolumn{11}{l}{\small\textit{Table \thetable\ continued from previous page}} \\
\toprule
& & \multicolumn{3}{c}{\textit{FIEG}} & \multicolumn{3}{c}{\textit{BIEG}} & \multicolumn{3}{c}{\textit{GIEG}} \\
\cmidrule(lr){3-5}\cmidrule(lr){6-8}\cmidrule(lr){9-11}
$T$ & $N$ & $0.01$ & $0.05$ & $0.10$ & $0.01$ & $0.05$ & $0.10$ & $0.01$ & $0.05$ & $0.10$ \\
\midrule
\endhead
\midrule
\multicolumn{11}{r}{\small\textit{continued on next page}} \\
\endfoot
\bottomrule
\endlastfoot
\multicolumn{11}{l}{\textit{Panel A: Case $\mathrm{c}$}} \\
\midrule
$90$ & $1$ & $0.021$ & $0.076$ & $0.148$ & $0.016$ & $0.059$ & $0.113$ & $0.011$ & $0.050$ & $0.144$ \\
$90$ & $3$ & $0.009$ & $0.057$ & $0.128$ & $0.011$ & $0.047$ & $0.115$ & $0.008$ & $0.058$ & $0.120$ \\
$140$ & $1$ & $0.010$ & $0.051$ & $0.097$ & $0.009$ & $0.047$ & $0.100$ & $0.008$ & $0.046$ & $0.079$ \\
$140$ & $3$ & $0.010$ & $0.044$ & $0.090$ & $0.010$ & $0.047$ & $0.091$ & $0.009$ & $0.051$ & $0.099$ \\
$350$ & $1$ & $0.011$ & $0.049$ & $0.102$ & $0.009$ & $0.051$ & $0.102$ & $0.009$ & $0.048$ & $0.108$ \\
$350$ & $3$ & $0.011$ & $0.052$ & $0.103$ & $0.008$ & $0.048$ & $0.095$ & $0.008$ & $0.050$ & $0.106$ \\
$600$ & $1$ & $0.011$ & $0.052$ & $0.108$ & $0.014$ & $0.058$ & $0.105$ & $0.011$ & $0.050$ & $0.102$ \\
$600$ & $3$ & $0.012$ & $0.051$ & $0.096$ & $0.008$ & $0.049$ & $0.104$ & $0.010$ & $0.049$ & $0.100$ \\
\midrule
\multicolumn{11}{l}{\textit{Panel B: Case $\mathrm{ct}$}} \\
\midrule
$90$ & $1$ & $0.023$ & $0.076$ & $0.143$ & $0.014$ & $0.062$ & $0.122$ & $0.009$ & $0.061$ & $0.127$ \\
$90$ & $3$ & $0.012$ & $0.057$ & $0.108$ & $0.016$ & $0.057$ & $0.101$ & $0.008$ & $0.048$ & $0.106$ \\
$140$ & $1$ & $0.012$ & $0.049$ & $0.091$ & $0.006$ & $0.038$ & $0.086$ & $0.008$ & $0.044$ & $0.095$ \\
$140$ & $3$ & $0.012$ & $0.050$ & $0.103$ & $0.012$ & $0.054$ & $0.104$ & $0.010$ & $0.046$ & $0.088$ \\
$350$ & $1$ & $0.007$ & $0.049$ & $0.095$ & $0.008$ & $0.050$ & $0.102$ & $0.010$ & $0.052$ & $0.102$ \\
$350$ & $3$ & $0.013$ & $0.057$ & $0.108$ & $0.008$ & $0.046$ & $0.099$ & $0.011$ & $0.049$ & $0.104$ \\
$600$ & $1$ & $0.010$ & $0.051$ & $0.101$ & $0.010$ & $0.051$ & $0.103$ & $0.010$ & $0.051$ & $0.100$ \\
$600$ & $3$ & $0.010$ & $0.045$ & $0.093$ & $0.008$ & $0.052$ & $0.099$ & $0.010$ & $0.053$ & $0.105$ \\
\end{longtable}
\vspace{-7mm}\footnotesize \textbf{Note:} Entries are empirical rejection frequencies of the $p$-value from Equation~\eqref{eq:ndf} under the null, and should equal the nominal level in the column heading. The Monte Carlo standard error of an entry is $0.002$ at the $1\%$ level, $0.004$ at $5\%$, and $0.005$ at $10\%$. $T = 90$ lies below the smallest sample size used to fit the surfaces and is included to show where the approximation fails.
\normalsize\vspace{5.5mm}

At $T = 350$ and $600$ the calibration is close to exact: the $1\%$ level returns between $0.007$ and $0.014$, the $5\%$ level between $0.045$ and $0.058$, and the $10\%$ level between $0.093$ and $0.108$, with two of $72$ entries more than two Monte Carlo standard errors from nominal. At $T = 140$ the fit is looser, and the misses are one-sided: five of $36$ entries fall more than two standard errors below their nominal level and none above, the largest being \textit{GIEG} at $10\%$ in Case $\mathrm{c}$ ($0.079$ against $0.100$) and \textit{BIEG} at $5\%$ in Case $\mathrm{ct}$ ($0.038$ against $0.050$). However, \textit{FIEG} is within two standard errors throughout. The tabulation is therefore mildly conservative at the sample size of Section~\ref{sec:application}, so the rejections reported there are, if anything, understated.

At $T = 90$ the tests over-reject by up to $0.076$ against a nominal $0.05$ for \textit{FIEG}. That sample size lies below the smallest in the fitting grid, so Equation~\eqref{eq:surfaceresponse} is extrapolated rather than interpolated, and the finite-sample terms diverge quickly below the fitted range.

The tabulation is therefore valid for $T \geq 100$, for $N \leq 3$, for the two deterministic cases of Equation~\eqref{eq:detcases}, and for $r_0 = 0.15$, a restriction that binds Section~\ref{sec:application} as well. The first three are boundaries of the present grid. However, the trimming fraction is different in kind, since it enters the definition of the statistic and so changes the null distribution rather than merely the point at which it is evaluated.

\subsection{Practical Guide: Choosing Between the Tests}\label{sec:practicalguide}

The choice among $\textit{FIEG}$, $\textit{BIEG}$, and $\textit{GIEG}$ should be guided by what the practitioner already assumes about the location of a possible break. $\textit{FIEG}$ fixes the start of the window at the beginning of the sample and varies the endpoint, which makes it a natural choice when the cointegrating relationship is believed to hold from the outset and to break down later. $\textit{BIEG}$ fixes the endpoint and varies the starting point, and is therefore more appropriate when the relationship is expected to hold currently, but not necessarily at the beginning of the sample. $\textit{GIEG}$ makes neither assumption and varies both points jointly, so it is preferred when the location of a break is unknown a priori. The flexibility of $\textit{GIEG}$ carries a computational cost, since it examines many more windows. Moreover, by Proposition~\ref{prop:consistency}, $\textit{FIEG}$ has no power against a relationship that begins part-way through the sample, and $\textit{BIEG}(r_0,1)$ has none against one that ends part-way through. Choosing the wrong direction eliminates power rather than merely reducing it. We therefore recommend $\textit{FIEG}$ or $\textit{BIEG}$ only where directional prior knowledge is genuinely available, and $\textit{GIEG}$ whenever the direction is in doubt.

The tests are also naturally applied recursively, and we would stress this as the way $\textit{GIEG}$ is intended to be used. A rejection locates a single window, not the only one. Having found one window, a practitioner should remove it and re-run $\textit{GIEG}$ on each of the segments that remain, repeating until no segment rejects or the window sizes become too short, so that the record is mapped into the sub-periods over which a long-run relationship holds. The critical values tabulated in Section~\ref{sec:responsesurface} apply to the first stage of such a scheme only, since every later segment is selected by the outcome of the stage before it; Section~\ref{sec:application} returns to this point.

\section{Application: Temperature and Sea Level}\label{sec:application}

The relationship between global sea level and surface temperature has been a subject of considerable scientific debate, particularly regarding appropriate statistical methodology. While semi-empirical models have been proposed to estimate this relationship, earlier approaches did not properly account for the dependence structure of the data \citep{Rahmstorf2007ARise}. This oversight has significant implications: when one non-stationary process is regressed on another, spurious inference is likely to be drawn \citep{Schmith2007CommentRise}.

The application of cointegration methodology is, therefore, essential for this relationship. Cointegration analysis is capable of handling the peculiarities inherent in trending time series, providing a statistically rigorous framework for understanding how sea level and temperature move together. \citet{Schmith2012StatisticalMethods} tested for a cointegration relationship between global mean sea level and global mean surface temperature. They uncover that the two series share a single common stochastic trend, meaning they have cointegration rank one. When they decompose their fitted vector error-correction model into its two single-equation adjustment relations, they find that temperature adjusts significantly to deviations from the long-run relationship, while the coefficient for sea level is not statistically significant. In the EG single-equation framework, this corresponds to treating temperature as the $y$ variable and sea level as the $x$ variable. 

Besides uncovering the significant cointegrating relationship, \citet{Schmith2012StatisticalMethods} also find that the disequilibrium error is not stable throughout the sample. This provides a natural motivating application for Cointegration by Parts tests. Rather than assuming the cointegration holds uniformly across the full sample, our tests scan multiple subsamples for the strongest evidence of cointegration, and thereby indicate where in the sample that evidence is concentrated.

We use the CSIRO reconstruction of global mean sea level [GMSL], following the Church and White methodology used by \citet{Schmith2012StatisticalMethods}, and the HadCRUT5 global mean surface temperature [GMST] anomaly series \citep{Legresy2021GlobalV1., Morice2021AnSet}. Both series are aggregated to yearly frequency and aligned, given that the CSIRO data ends in 2019, leaving us with $T = 140$ observations from 1880--2019. Further, the HadCRUT5 data is rebased from the 1961--1990 reference period to the standard 1850--1900 pre-industrial baseline. This rebasing is immaterial for the test, since the constant in Equation~\eqref{eq:model1} absorbs any change of baseline in either series; we apply it only so that the plotted levels are directly interpretable against the usual pre-industrial reference. The two time series are shown in Figure~\ref{fig:GMSL_GMTA}.

\begin{figure}[H]
    \centering
    \includegraphics[width=.8\linewidth]{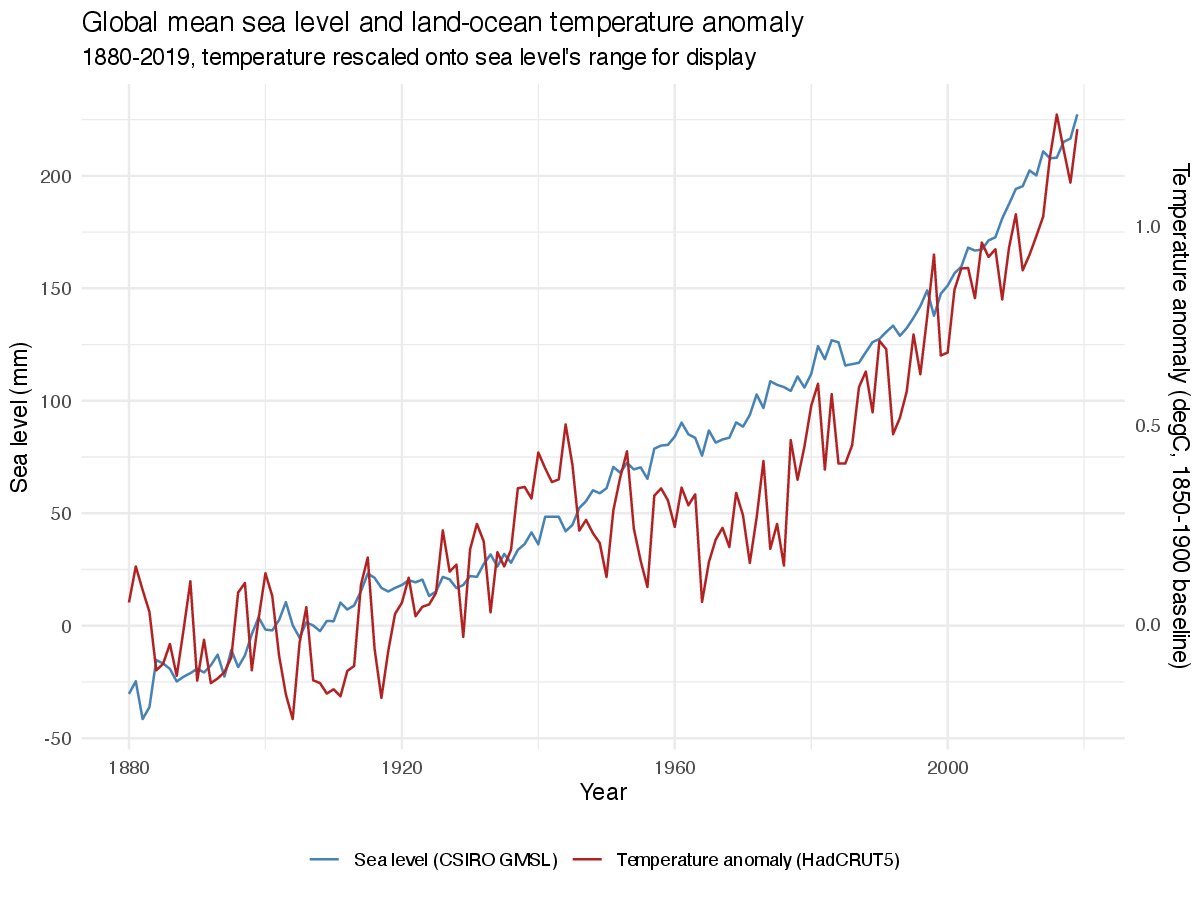}
    \caption{Yearly series of global mean sea level (measured in mm, blue line) and global surface temperature anomaly (measured in $^{\circ}$C, red line).}
    \label{fig:GMSL_GMTA}
\end{figure}

Before we can apply our Cointegration by Parts tests, we need to determine if GMSL and temperature anomalies are non-stationary. We perform ADF tests on the levels and first differences of both series, with the lag order selected by BIC as in Section~\ref{sec:simdesign}. In levels, neither series gives any evidence against a unit root, under either deterministic specification: the statistics are $2.21$ for sea level and $0.64$ for temperature with a constant, and $-0.42$ and $-1.35$ with a constant and a trend, against $5\%$ critical values of approximately $-2.88$ and $-3.44$ at $T=140$. In first differences the null is rejected decisively for both, with statistics of $-15.59$ and $-10.97$. Both series are therefore treated as $I(1)$. Following the finding of \citet{Schmith2012StatisticalMethods}, we let temperature be $y$ and sea level $x$, as the single regressor. We perform the $\textit{FIEG}(r_0)$, $\textit{BIEG}(r_0, r_2)$, and $\textit{GIEG}(r_0)$ tests using $r_0 = 0.15$, which is standard in the parameter-instability literature \citep{Andrews1993TestsPoint}, together with the lag rule of Section~\ref{sec:simdesign}. At $T = 140$ this makes the shortest admissible window $21$ observations, and $T$ is comfortably above the $T = 100$ floor established in Section~\ref{sec:ndf}.

Both series display a pronounced upward trend over the sample, which raises the question of the appropriate deterministic specification in Equation~\eqref{eq:detcases}. If $x_t$ and $y_t$ are integrated with drift, the drifts dominate the cointegrating regression under Case $\mathrm{c}$: the OLS estimator then converges to the ratio of the two drifts at rate $T^{3/2}$ rather than to the cointegrating parameter, and the Case $\mathrm{c}$ critical values, which are calibrated on driftless random walks, no longer apply. Including a linear trend, that is Case $\mathrm{ct}$, restores super-consistency of $\hat{\beta}$ and yields limits in terms of window-detrended Brownian motions; this is the distinction between deterministic and stochastic cointegration discussed by \citet{Hansen1992EfficientTrends}. We therefore treat Case $\mathrm{ct}$ as the appropriate specification for this application, and report Case $\mathrm{c}$ alongside it, tabulated under the same design in Section~\ref{sec:simulation}, as a comparison with the wider residual-based cointegration literature.

Table~\ref{tab:application_results} shows the three test statistics, the $p$-value from the numerical distribution function of Section~\ref{sec:ndf}, the corresponding response surface critical values, and the location of each extreme window, reported as calendar years.

\begin{table}[ht!]
    \centering
    \caption{\label{tab:application_results} Test statistics for the sea level and temperature relationship, $r_0 = 0.15$, $T = 140$, 1880--2019.}
    \footnotesize
    \begin{tabular}{l rr rrr r}
    \toprule
     & Statistic & $p$-value & \multicolumn{3}{c}{Critical value} & Window \\
     \cmidrule(lr){4-6}
     & & & $1\%$ & $5\%$ & $10\%$ & (Years) \\
\midrule
\multicolumn{7}{l}{\textit{Panel A: Case $\mathrm{c}$}} \\
\midrule
$\textit{FIEG}(r_0)$ & $-5.428$ & $0.011$ & $-5.479$ & $-4.663$ & $-4.273$ & 1880--2014 \\
$\textit{BIEG}(r_0,1)$ & $-5.812$ & $0.004$ & $-5.408$ & $-4.661$ & $-4.262$ & 1967--2019 \\
$\textit{GIEG}(r_0)$ & $-6.192$ & $0.090$ & $-7.452$ & $-6.461$ & $-6.152$ & 1970--2007 \\
\midrule
\multicolumn{7}{l}{\textit{Panel B: Case $\mathrm{ct}$}} \\
\midrule
$\textit{FIEG}(r_0)$ & $-5.620$ & $0.027$ & $-6.076$ & $-5.279$ & $-4.887$ & 1880--2019 \\
$\textit{BIEG}(r_0,1)$ & $-6.213$ & $0.009$ & $-6.141$ & $-5.306$ & $-4.914$ & 1967--2019 \\
$\textit{GIEG}(r_0)$ & $-7.414$ & $0.041$ & $-8.251$ & $-7.276$ & $-6.807$ & 1947--1967 \\
    \bottomrule
    \end{tabular}
    \\[4pt]
    \textbf{Note:} $N = 1$ regressor, lag order by BIC. Critical values are from Equation~\eqref{eq:surfaceresponse} with the coefficients of Table~\ref{tab:response_surface}; $p$-values are from the numerical distribution function of Equation~\eqref{eq:ndf}. Case $\mathrm{ct}$ is the specification argued for in the text.
\end{table}

Under Case $\mathrm{ct}$, all three statistics reject the null of no cointegration at the $5\%$ level, with $p$-values of $0.027$, $0.009$, and $0.041$. Under Case $\mathrm{c}$ the picture is similar for $\textit{FIEG}$ and $\textit{BIEG}$ but not for $\textit{GIEG}$, whose $p$-value of $0.090$ leaves it short of the $5\%$ level. This divergence is itself informative: the doubly-flexible statistic pays the largest critical-value penalty for its search, so it is the statistic most sensitive to using the wrong deterministic specification, and the Case $\mathrm{c}$ column is misspecified when both series drift.

The location of each window follows the pattern the construction of the statistics dictates. $\textit{FIEG}$ is constrained to begin at the start of the sample, and its extremum spans the whole period. $\textit{BIEG}$ is free to choose its own starting point and shifts to 1967--2019, discarding everything before the late 1960s. $\textit{GIEG}$, free in both endpoints, selects 1947--1967 under Case $\mathrm{ct}$.

The $\textit{GIEG}$ window is $21$ observations long, exactly the shortest the trimming admits at $r_0 = 0.15$ and $T = 140$. A minimizing window on the boundary of the search set makes the result sensitive to that choice, and Table~\ref{tab:application_robustness} confirms it: the $\textit{GIEG}$ $p$-value moves from $0.038$ at $r_0 = 0.10$ to $0.041$ at $r_0 = 0.15$ and $0.161$ at $r_0 = 0.20$. $\textit{FIEG}$ and $\textit{BIEG}$, whose windows are not at the boundary, are invariant to $r_0$ over this range. The $\textit{GIEG}$ rejection should therefore be read as suggestive rather than as firm evidence.

\begin{table}[H]
    \centering
    \caption{\label{tab:application_robustness} Robustness of the Case $\mathrm{ct}$ results to the trimming fraction and the normalization.}
    \footnotesize
    \setlength{\tabcolsep}{4pt}
    \begin{tabular}{cl rrc rrc rrc}
    \toprule
    & & \multicolumn{3}{c}{\textit{FIEG}} & \multicolumn{3}{c}{\textit{BIEG}} & \multicolumn{3}{c}{\textit{GIEG}} \\
    \cmidrule(lr){3-5}\cmidrule(lr){6-8}\cmidrule(lr){9-11}
    $r_0$ & Normalization & Stat. & $p$ & Window & Stat. & $p$ & Window & Stat. & $p$ & Window \\
    \midrule
$0.10$ & $y=$ temp. & $-5.620$ & $0.027$ & 1880--2019 & $-6.213$ & $0.009$ & 1967--2019 & $-7.457$ & $0.038$ & 1948--1965 \\
$0.10$ & $y=$ s.\,l. & $-4.858$ & $0.106$ & 1880--1893 & $-4.808$ & $0.120$ & 1993--2019 & $-5.477$ & $0.636$ & 1920--1999 \\
$0.15$ & $y=$ temp. & $-5.620$ & $0.027$ & 1880--2019 & $-6.213$ & $0.009$ & 1967--2019 & $-7.414$ & $0.041$ & 1947--1967 \\
$0.15$ & $y=$ s.\,l. & $-4.048$ & $0.359$ & 1880--1945 & $-4.808$ & $0.120$ & 1993--2019 & $-5.477$ & $0.636$ & 1920--1999 \\
$0.20$ & $y=$ temp. & $-5.620$ & $0.027$ & 1880--2019 & $-6.213$ & $0.009$ & 1967--2019 & $-6.484$ & $0.161$ & 1946--1973 \\
$0.20$ & $y=$ s.\,l. & $-4.048$ & $0.359$ & 1880--1945 & $-4.031$ & $0.345$ & 1992--2019 & $-5.477$ & $0.636$ & 1920--1999 \\
    \bottomrule
    \end{tabular}
    \\[4pt]
    \textbf{Note:} All entries are Case $\mathrm{ct}$ at $T = 140$, with the lag order selected by BIC. ``$y=$ temp.'' regresses temperature on sea level, which is the normalization implied by the weak-exogeneity finding of \citet{Schmith2012StatisticalMethods}; ``$y=$ s.\,l.'' reverses it. $p$-values are from Equation~\eqref{eq:ndf}.
\end{table}

Two features of the result deserve comment. First, the direction of the regression matters. With sea level rather than temperature on the left-hand side no statistic rejects at any conventional level, and the $\textit{GIEG}$ $p$-value rises to $0.64$. This is a known feature of single-equation residual-based testing rather than a defect of our procedure, and it is why the regressand is chosen here on the weak-exogeneity result of \citet{Schmith2012StatisticalMethods} rather than by convenience; a system analogue of the kind discussed in Section~\ref{sec:conclusion} would remove the choice altogether.

Second, and bearing more directly on what the result means, the window 1947--1967 sits on a known discontinuity in the observational record. \citet{Thompson2008ATemperature} document a spurious shift in global sea surface temperature centered on 1945, caused by a wartime change in measurement methodology and corroborated by \citet{Chan2019CorrectingWarming, Chan2019SystematicMethod}. The rejection may therefore reflect a genuine episode of tighter co-movement, or a measurement artifact that produces a mean-reverting residual over the same stretch. These two readings cannot be separated with the statistics reported here. What the procedure does deliver is the location: it identifies the stretch of the record over which the question has to be settled, which a full-sample test would average away.

We illustrate the GMSL and GMST anomalies, with the shaded area indicating the window that minimizes the $\textit{GIEG}(r_0)$ statistic under Case $\mathrm{ct}$, in Figure~\ref{fig:windows_application}.

\begin{figure}[ht!]
    \centering
    \includegraphics[width=0.8\linewidth]{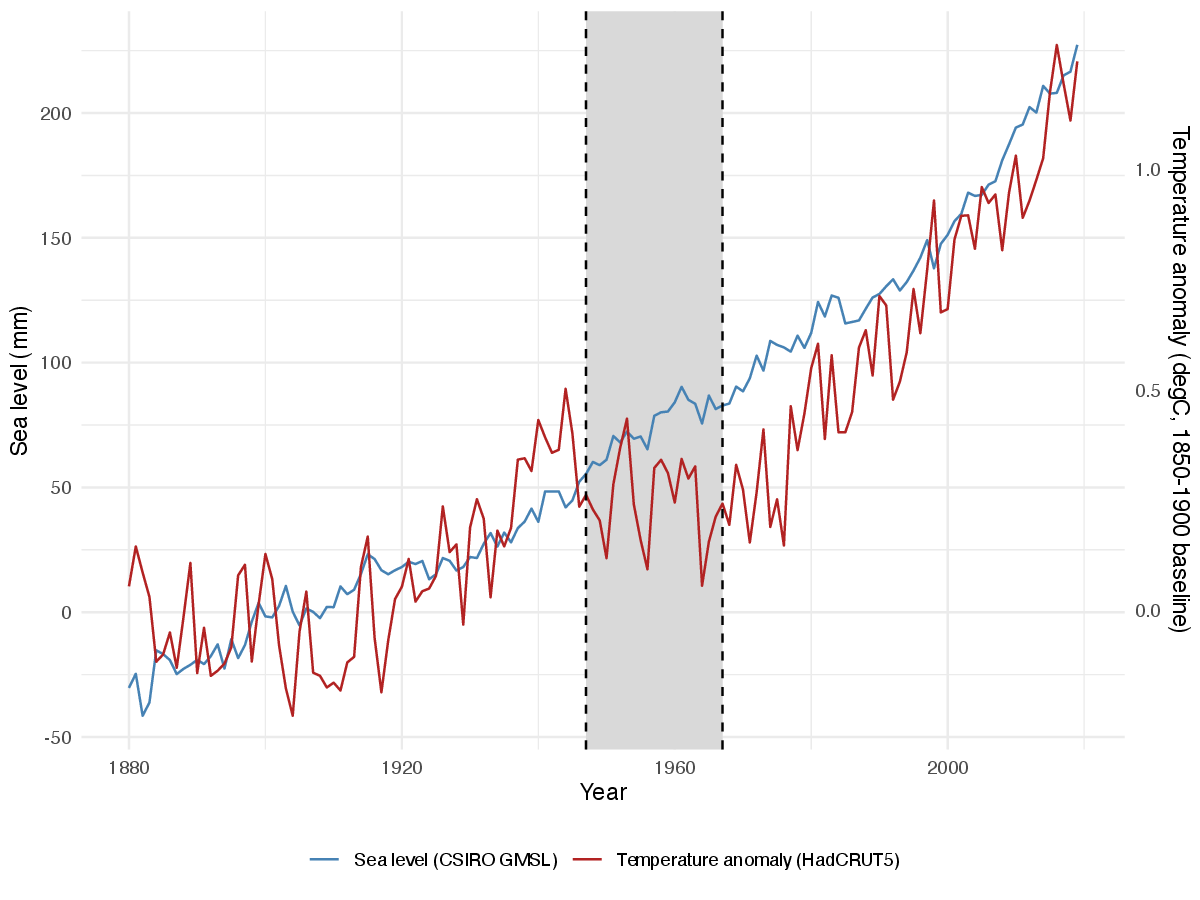}
    \caption{Yearly series of global mean sea level (measured in mm, blue line) and global mean surface temperature anomaly (measured in $^{\circ}$C, red line). The dashed vertical black lines mark the start and end of the window that minimizes the $\textit{GIEG}(r_0)$ statistic under Case $\mathrm{ct}$; that window is shaded in gray.}
    \label{fig:windows_application}
\end{figure}

As set out in Section~\ref{sec:practicalguide}, the natural next step after locating a window is to remove it and re-run the tests, and in particular $\textit{GIEG}$, on the segments that remain, asking whether a further episode of cointegration is present elsewhere in the record. The sample here does not permit it: removing 1947--1967 leaves segments of $67$ and $52$ observations, both well below the $T = 100$ floor of Section~\ref{sec:ndf}. 

Our results can be used as suggestive evidence that the long-run relationship between sea level and global mean surface temperature anomalies is concentrated in part of the record rather than uniform across it. They cannot, by themselves, be used to conclude that cointegration genuinely breaks down outside the located window: as set out in Section~\ref{sec:intro}, the null here is no cointegration anywhere, so a rejection localizes evidence and does not establish a breakdown. Nor can they establish that the located window reflects the climate system rather than the instrument record, for the reasons given above. What they do establish is that a full-sample Engle-Granger test is not the only reading available, and that the evidence in this well-studied relationship is not spread evenly across the 1880--2019 period.

\section{Conclusion}\label{sec:conclusion}

There are several real-world scenarios in which cointegration cannot be assumed to hold uniformly across a sample. The restriction that a single window spanning the whole record is the right object to test has been called into question by, e.g., \citet{Gregory1996Residual-basedShifts, Bierens2010TIME-VARYINGCOINTEGRATION}. \citet{Gregory1996Residual-basedShifts} was among the first to examine various types of structural breaks in cointegration, while \citet{Bierens2010TIME-VARYINGCOINTEGRATION} found evidence of time-varying cointegration between international prices and nominal exchange rates. Thus, there is a need for testing procedures that do not impose constant co-movement between series over the full sample. We build on the methodology developed for testing for financial bubbles, in which \citet{Phillips2011BehaviorValues} construct a test statistic from standard unit root machinery, and on \citet{Gregory1996Residual-basedShifts} for the use of an infimum-based test.

We derive the limiting distributions of the three statistics, show that they are free of nuisance parameters, and tabulate them by simulation in both deterministic specifications. From those simulations, we construct a numerical distribution function in the sense of \citet{MacKinnon1994ApproximateTests, MacKinnon1996NumericalTests}, which converts any observed statistic at any sample size into an approximate $p$-value, and which we verify is correctly calibrated at sample sizes outside the grid on which it was fitted. Proposition~\ref{prop:consistency} establishes which alternatives each statistic is consistent against: $\textit{FIEG}$ has power only when the cointegrating episode begins at the start of the sample, $\textit{BIEG}(r_0,1)$ only when it runs to the end, and $\textit{GIEG}$ in either case, at a greater computational cost.

Lastly, we apply our tests to sea level and global mean surface temperature anomalies, following the results of \citet{Schmith2012StatisticalMethods}. Under the constant-plus-trend specification that the drifting series require, all three statistics reject the null of no cointegration at the $5\%$ level, and the minimizing windows place the evidence in identifiable sub-periods rather than uniformly across the record. The window selected by $\textit{GIEG}$, 1947--1967, coincides with the bucket-to-intake discontinuity in the sea surface temperature record documented by \citet{Thompson2008ATemperature}.

Three limitations should be kept in view, and each points to a natural extension. First, the hypotheses in Equation~\eqref{eq:hypotheses} pair a null of no cointegration anywhere against an alternative of cointegration somewhere. Therefore, a rejection localizes the evidence for a long-run relationship, it does not establish that a relationship which held earlier has since broken down. Testing that claim requires the complementary pairing, with cointegration over the full sample as the null. Second, the procedure is single-equation and residual-based, so it detects at most one cointegrating relationship and inherits the dependence on the choice of regressand documented in Section~\ref{sec:application}; a system analogue, in which the rank of the cointegrating space is tested window by window in the manner of \citet{Johansen1988StatisticalVectors, Johansen1990MaximumMoney, Johansen1991EstimationModels, Johansen1995Likelihood-BasedModels}, would relax both restrictions. Ongoing work is developing these two extensions. Third, because the cointegrating vector is re-estimated within each window, a rejection is equally consistent with a relationship that holds throughout the sample but whose coefficients shift at some date, as in \citet{Gregory1996Residual-basedShifts}. Distinguishing partial cointegration from a shift in the cointegrating vector requires additional structure; as a simple first diagnostic, we recommend reporting the window-by-window estimates of $\hat{\beta}$ alongside the test statistics.

\section*{Reproducibility}

The code used for the simulation study and empirical application in Section~\ref{sec:simulation}--\ref{sec:application} can be found on GitHub (\href{https://github.com/okvist04/Cointegration-by-Parts-Econometric-Reviews-Replication}{Reproducibility Repository}). It contains an implementation of $\textit{FIEG}$, $\textit{BIEG}$, and $\textit{GIEG}$ together with the numerical distribution function of Section~\ref{sec:ndf}. That means, the $p$-values can be obtained directly, the response surface coefficients at all $184$ probabilities of the grid $\mathcal{P}$, in both deterministic cases and for every $N$, of which Table~\ref{tab:response_surface} prints only the three conventional levels. Further, the simulated minima from which they are estimated and the scripts that produce every table and figure in Sections~\ref{sec:simulation} and~\ref{sec:application} can also be obtained here.

\section*{Acknowledgments}
We thank the participants of the 9th Conference on Econometric Models of Climate Change in 2025 for their useful comments and suggestions. We extend special thanks to Felix Pretis and Charisios Grivas for their valuable insights and suggestions. 

\section*{Declaration of Interest}
The authors declare that they have no conflict of interest. 

\section*{Funding}

This research did not receive any specific grant from funding agencies in the public, commercial, or not-for-profit sectors.

\section*{CRediT Authorship Contribution Statement}

\textbf{Olivia Kvist:} Conceptualization, Methodology, Data curation, Software, Formal analysis, Investigation, Writing - original draft, Writing - review \& editing, Visualization.

\textbf{J. Eduardo Vera-Valdés:} Conceptualization, Methodology, Writing - review \& editing, Supervision, Validation. 

\section*{Declaration of Generative AI}

During the preparation of this paper, the authors utilized generative AI to correct typos and speed up computations. After using generative AI, the authors reviewed and edited the content as needed and take full responsibility for the content of the published article.

\section*{Data Availability Statement}

The data used in Section~\ref{sec:application} is freely available. The GMSL can be found on \href{http://hdl.handle.net/102.100.100/482857?index=1}{CSIRO} and the global surface temperature can be found on \href{https://www.metoffice.gov.uk/hadobs/hadcrut5/data/HadCRUT.5.1.0.0/download.html}{HadCRUT5}.

\bibliographystyle{elsarticle-harv}
\bibliography{references}

\appendix

% Numbering in the appendix. In the body the theorem-like counters print as
% \arabic{section}.<n>; elsarticle sets \thesection to "Appendix A" here, so \arabic{section}
% would restart at 1. Switch to the section letter and reset each counter at every appendix
% section, so that the results print as A.1, A.2, ..., B.1, ..., as the text refers to them.
\numberwithin{thm}{section}
\numberwithin{lemma}{section}
\numberwithin{proposition}{section}
\numberwithin{remark}{section}
\renewcommand{\thethm}{\Alph{section}.\arabic{thm}}
\renewcommand{\thelemma}{\Alph{section}.\arabic{lemma}}
\renewcommand{\theproposition}{\Alph{section}.\arabic{proposition}}
\renewcommand{\theremark}{\Alph{section}.\arabic{remark}}

\section{Limiting Objects, Supporting Lemmas, and Proofs}\label{app:theory}

This appendix collects the material supporting Theorem~\ref{thm:FIEG}: the construction of the
limiting objects, the lemmas on which all three statistics rest, the proof of the theorem itself,
the invariance of its limit to nuisance parameters, the consistency of the three search sets against
the two break directions, and the treatment of data-dependent lag selection.

It may help to have the roles of the four lemmas in view before the details. Lemma~\ref{lem:gram}
says that the window regressions never degenerate, and says it at every window at once rather than
at one window at a time. Lemma~\ref{lem:uniform} then upgrades the convergence of
$\textit{EG}(r_1,r_2)$ at a single window to convergence of the whole process indexed by the window,
which is what allows the infimum to be taken in the limit. Those two carry the null hypothesis. The
remaining two describe the alternative: Lemma~\ref{lem:clean} shows that a window lying inside the
stationary segment sends the statistic to $-\infty$, and Lemma~\ref{lem:tight} that a family of
windows held away from that segment stays bounded. Proposition~\ref{prop:consistency} then reads off
which of the two applies to each of the three search sets. A reader willing to take the machinery on
trust may go directly from Equation~\eqref{eq:Qfunctional}, which names the limit, to
Proposition~\ref{prop:consistency}.

Throughout, $\textit{EG}(r_1,r_2)$ is the window statistic of Section~\ref{sec:theory}, written out here in
full. The estimator of Equation~\eqref{eq:recursive_beta} and the residuals of
Equation~\eqref{eq:recursive_residual} are computed on $[\lfloor Tr_1\rfloor+1,\,\lfloor Tr_2\rfloor]$ in place
of $[1,\lfloor Tr\rfloor]$,
\begin{align}\label{eq:beta_window}
    \hat{\delta}_{r_1,r_2} = \begin{pmatrix}\hat{\theta}_{r_1,r_2} \\ \hat{\beta}_{r_1,r_2}\end{pmatrix}
      & = \left( \sum_{t=\lfloor Tr_1\rfloor+1}^{\lfloor Tr_2\rfloor} z_t z_t^\top \right)^{-1}
          \sum_{t=\lfloor Tr_1\rfloor+1}^{\lfloor Tr_2\rfloor} z_t y_t , \\
    \hat{\varepsilon}_{t,r_1,r_2} & = y_t - \hat{\theta}_{r_1,r_2}^\top d_t - \hat{\beta}_{r_1,r_2}^\top x_t ,
      \nonumber
\end{align}
for $(r_1,r_2)\in\mathcal{R}(r_0)$; the ADF regression of Equation~\eqref{eq:residmodel} is run on
$\hat{\varepsilon}_{t,r_1,r_2}$ over the same window, and
$\textit{EG}(r_1,r_2) := \hat{\gamma}_{r_1,r_2}/\operatorname{se}(\hat{\gamma}_{r_1,r_2})$, the recursive case
being $r_1 = 0$.

\subsection{The Limiting Objects}\label{app:limitobjects}

Because the estimator in Equation~\eqref{eq:recursive_beta} includes the deterministic regressors, the limiting objects are built from the residuals of a projection on $d$ rather than from the Brownian motions themselves. Let $d(\cdot)$ denote the limit of the deterministic regressors, so that $d(s) = 1$ in Case $\mathrm{c}$ and $d(s) = (1,s)^\top$ in Case $\mathrm{ct}$. For a window $(r_1,r_2)\in\mathcal{R}(r_0)$ and any scalar process $B$ on $[r_1,r_2]$, define the $L^2[r_1,r_2]$-projection residual
\begin{equation}\label{eq:projection}
    B^{\perp}(s) := B(s) - d(s)^\top \left[\int_{r_1}^{r_2} d(u) d(u)^\top \mathrm{d}u \right]^{-1} \int_{r_1}^{r_2} d(u) B(u)\, \mathrm{d}u, \qquad s\in[r_1,r_2],
\end{equation}
suppressing the dependence on $(r_1,r_2)$ in the notation, and applying Equation~\eqref{eq:projection} componentwise when the process is vector-valued, as it is for $B_x$.

In Case $\mathrm{n}$, $B^{\perp} = B$; in Case $\mathrm{c}$, $B^{\perp}$ is $B$ demeaned over the window; and in Case $\mathrm{ct}$, $B^{\perp}$ is $B$ demeaned and detrended over the window. Equation~\eqref{eq:projection} is the continuous-time counterpart of partialling $d_t$ out of the regression, so that the slope estimator in Equation~\eqref{eq:recursive_beta} is the one obtained by regressing projected $y$ on projected $x$.

Under Assumption~\ref{ass:fclt}, define for every admissible window $(r_1,r_2)\in\mathcal{R}(r_0)$
\begin{align}\label{eq:limit_objects}
    b_{r_1,r_2} & := \left[\int_{r_1}^{r_2} B_x^{\perp}(u) B_x^{\perp}(u)^\top \mathrm{d}u \right]^{-1} \int_{r_1}^{r_2} B_x^{\perp}(u) B_y^{\perp}(u)\, \mathrm{d}u, \\
    B_{\hat{\varepsilon},r_1,r_2}(s) & := B_y^{\perp}(s) - B_x^{\perp}(s)^\top b_{r_1,r_2}, \qquad s \in [r_1,r_2], \nonumber \\
    \omega_{\hat{\varepsilon},r_1,r_2}^{2} & := \omega_{yy} - 2\, b_{r_1,r_2}^\top \Omega_{xy} + b_{r_1,r_2}^\top \Omega_{xx}\, b_{r_1,r_2} . \nonumber
\end{align}
Setting $a_{r_1,r_2} := (-b_{r_1,r_2}^\top,\, 1)^\top$, the variance rate is the quadratic form $\omega_{\hat{\varepsilon},r_1,r_2}^{2} = a_{r_1,r_2}^\top \Omega\, a_{r_1,r_2}$, so that
\begin{equation}\label{eq:omegabound}
    \omega_{\hat{\varepsilon},r_1,r_2}^{2} \;\geq\; \lambda_{\min}(\Omega)\,\lVert a_{r_1,r_2}\rVert^{2} \;\geq\; \lambda_{\min}(\Omega) \;>\; 0
\end{equation}
at every window and every realization, the second inequality because the last entry of $a_{r_1,r_2}$ is one and the third by the positive definiteness of $\Omega$ in Assumption~\ref{ass:fclt}(i). Equation~\eqref{eq:omegabound} is the bound used in Step~3 of the proof of Lemma~\ref{lem:uniform}. The variance rate takes the same form in all three cases of Equation~\eqref{eq:detcases}, because the projection in Equation~\eqref{eq:projection} subtracts a process of finite variation and therefore leaves quadratic variation unchanged.

Let $W_{\hat{\varepsilon},r_1,r_2} := B_{\hat{\varepsilon},r_1,r_2}/\omega_{\hat{\varepsilon},r_1,r_2}$ denote the standardization of $B_{\hat{\varepsilon},r_1,r_2}$ by its quadratic-variation rate $\omega_{\hat{\varepsilon},r_1,r_2}$. For the recursive windows of Theorem~\ref{thm:FIEG} we abbreviate $b_r := b_{0,r}$, $B_{\hat{\varepsilon},r} := B_{\hat{\varepsilon},0,r}$, $\omega_{\hat{\varepsilon},r} := \omega_{\hat{\varepsilon},0,r}$ and $W_{\hat{\varepsilon},r} := W_{\hat{\varepsilon},0,r}$. The normalization by $\omega_{\hat{\varepsilon},r_1,r_2}$ is required because the ratios appearing in Theorems~\ref{thm:FIEG}, \ref{thm:BIEG}, and \ref{thm:GIEG} are homogeneous of degree one in their argument process, whereas the $t$-statistics they approximate are scale-free.

It is convenient to name the limit itself. For $(r_1,r_2)\in\mathcal{R}(r_0)$ define
\begin{equation}\label{eq:Qfunctional}
    \mathcal{Q}(r_1,r_2) := \frac{W_{\hat{\varepsilon},r_1,r_2}(r_2)^2 - W_{\hat{\varepsilon},r_1,r_2}(r_1)^2 - (r_2-r_1)}{2\left\{\int_{r_1}^{r_2} W_{\hat{\varepsilon},r_1,r_2}(s)^2\,\mathrm{d}s\right\}^{1/2}}.
\end{equation}
Equation~\eqref{eq:Qfunctional} is the form in which the Dickey-Fuller limit is usually written for a
process of unit variance rate, and it is the form used throughout the recursive unit-root literature
\citep{Phillips2011BehaviorValues, Phillips2015TestingSP500}.

Writing the limit this way rather than as the
stochastic integral $\int_{r_1}^{r_2} W_{\hat{\varepsilon},r_1,r_2}\,\mathrm{d}W_{\hat{\varepsilon},r_1,r_2}$
is not a matter of taste here, and the reason is specific to our setting. The coefficient
$b_{r_1,r_2}$ in Equation~\eqref{eq:limit_objects} is a functional of the entire Brownian path over
$[r_1,r_2]$, so $W_{\hat{\varepsilon},r_1,r_2}$ is not adapted to the filtration generated by
$(B_x,B_y)$ and the stochastic integral is not an It\^{o} integral. No such difficulty arises for
Equation~\eqref{eq:Qfunctional}: the numerator involves only two endpoint values and the quadratic
variation, and $\langle W_{\hat{\varepsilon},r_1,r_2}\rangle_s = s$ because, by construction,
$W_{\hat{\varepsilon},r_1,r_2}$ is $B_{\hat{\varepsilon},r_1,r_2}$ divided by its own quadratic-variation
rate, and the projection in Equation~\eqref{eq:projection} subtracts a process of finite variation and
therefore leaves quadratic variation unchanged. The two expressions agree whenever both are defined,
by the pathwise It\^{o} formula, which holds for any continuous path possessing a quadratic variation
and requires no adaptedness or probabilistic structure \citep{Follmer1981CalculProbabilites}.

We therefore take Equation~\eqref{eq:Qfunctional} as the definition. Two things follow that we use below:
the limit is well defined for every window, and it is a continuous function of $(r_1,r_2)$ on
$\mathcal{R}(r_0)$, which is what the infimum in each of the three statistics requires.

\subsection{Supporting Lemmas}\label{app:uniform}

The window regressions are only well behaved if the limiting regressor Gram matrix stays away from
singularity, and because the statistics minimize over windows, it must do so at every window at once rather
than at each window separately. That is stronger than it sounds: the minimizing window is random, so
almost sure non-singularity for each fixed $(r_1,r_2)$ does not by itself rule out degeneracy somewhere on
$\mathcal{R}(r_0)$. It is nevertheless a consequence of Assumption~\ref{ass:fclt}(i) alone.

\begin{lemma}[Uniform non-degeneracy of the window Gram matrix]\label{lem:gram}
    Let Assumption~\ref{ass:fclt}(i) hold and let $r_0\in(0,1)$. Then, in each of the three cases of
    Equation~\eqref{eq:detcases},
    \begin{equation}\label{eq:uniformgram}
        \inf_{(r_1,r_2)\in\mathcal{R}(r_0)} \lambda_{\min}\!\left( \int_{r_1}^{r_2} B_x^{\perp}(u) B_x^{\perp}(u)^\top \mathrm{d}u \right) > 0 \qquad \text{almost surely},
    \end{equation}
    where $\lambda_{\min}$ denotes the smallest eigenvalue and $B_x^{\perp}$ is defined in
    Equation~\eqref{eq:projection}.
\end{lemma}

\begin{proof}[Proof of Lemma~\ref{lem:gram}]
    Write
    \begin{equation*}
        M(r_1,r_2) := \int_{r_1}^{r_2} B_x^{\perp}(u)B_x^{\perp}(u)^\top \mathrm{d}u, \qquad
        G(r_1,r_2) := \int_{r_1}^{r_2} d(u)d(u)^\top \mathrm{d}u,
    \end{equation*}
    and set $L := r_2 - r_1 \geq r_0$. The argument has three steps.

    \emph{Step 1. The deterministic Gram matrix is uniformly invertible.} In Case $\mathrm{n}$ there is nothing to
    prove, since $B_x^{\perp} = B_x$. In Case $\mathrm{c}$, $G = L \geq r_0$. In Case $\mathrm{ct}$ a direct
    computation gives $\det G = L^{4}/12$, while $\operatorname{tr} G = L + (r_2^3-r_1^3)/3 \leq 2$ for
    $r_1,r_2\in[0,1]$, so that
    $\lambda_{\min}(G) = \det G / \lambda_{\max}(G) \geq \det G / \operatorname{tr} G \geq r_0^{4}/24$.
    Hence $\sup_{\mathcal{R}(r_0)}\lVert G^{-1}\rVert < \infty$ and $(r_1,r_2)\mapsto G(r_1,r_2)^{-1}$ is
    continuous on $\mathcal{R}(r_0)$. This is the first of the two places where the trimming is used.

    \emph{Step 2. Continuity and compactness.} On the almost sure event that $B_x$ has continuous sample paths, the
    projection coefficient $C(r_1,r_2) := G(r_1,r_2)^{-1}\int_{r_1}^{r_2} d(u)B_x(u)^\top \mathrm{d}u$, read as $C \equiv 0$ in Case $\mathrm{n}$, is
    continuous on $\mathcal{R}(r_0)$ by Step~1, hence so is
    $B_x^{\perp}(\cdot) = B_x(\cdot) - C(r_1,r_2)^\top d(\cdot)$ jointly in its argument and in the window, and
    hence so is $M$. The set $\mathcal{R}(r_0)$ is compact and $\lambda_{\min}$ is a continuous function of a
    symmetric matrix, so $\lambda_{\min}(M(\cdot,\cdot))$ attains its infimum on $\mathcal{R}(r_0)$. It is
    therefore enough to show that the infimum is attained at a positive value, that is, that
    $\lambda_{\min}(M(r_1,r_2))>0$ at every window, on a single almost sure event.

    \emph{Step 3. Positivity at every window simultaneously.} Let $\mathcal{N}$ denote the almost sure event on which
    $B_x$ has continuous sample paths and quadratic variation $\langle B_x\rangle_u = u\,\Omega_{xx}$ for all
    $u\in[0,1]$. Fix $\omega\in\mathcal{N}$, a window $(r_1,r_2)\in\mathcal{R}(r_0)$ and a vector
    $a\in\mathbb{R}^N$ with $\lVert a\rVert = 1$, and suppose $a^\top M(r_1,r_2) a = 0$. Then
    $\int_{r_1}^{r_2}\{a^\top B_x^{\perp}(u)\}^2\mathrm{d}u = 0$, and since the integrand is continuous,
    $a^\top B_x^{\perp} \equiv 0$ on $[r_1,r_2]$. Equivalently,
    \begin{equation*}
        a^\top B_x(u) = \{C(r_1,r_2)a\}^\top d(u), \qquad u \in [r_1,r_2],
    \end{equation*}
    so that $a^\top B_x$ agrees on $[r_1,r_2]$ with an element of the span of $d$: the zero function in
    Case $\mathrm{n}$, a constant in Case $\mathrm{c}$, and an affine function in Case $\mathrm{ct}$. Such a function is of bounded variation and,
    being continuous, has zero quadratic variation on $[r_1,r_2]$. On $\mathcal{N}$, however, the quadratic
    variation of $a^\top B_x$ over $[r_1,r_2]$ equals
    \begin{equation*}
        a^\top\left\{\langle B_x\rangle_{r_2} - \langle B_x\rangle_{r_1}\right\}a = (r_2-r_1)\, a^\top \Omega_{xx} a \; \geq \; r_0\, \lambda_{\min}(\Omega_{xx}) \; > \; 0,
    \end{equation*}
    the last inequality by the positive definiteness of $\Omega$ in Assumption~\ref{ass:fclt}(i). This is a
    contradiction, so $a^\top M(r_1,r_2)a>0$ for every unit $a$, that is
    $\lambda_{\min}(M(r_1,r_2))>0$.

    The quadratic-variation identity is a single statement about the path of $B_x$, so it holds on
    $\mathcal{N}$ simultaneously for every window and every direction; this is what removes the
    difficulty that the minimizing window is random. Combining with Step~2,
    Equation~\eqref{eq:uniformgram} holds on $\mathcal{N}$.
\end{proof}

The argument uses only the quadratic variation of the limiting process, so it holds for every $N$ and in all three
deterministic cases without modification, and it needs no bound on how close to singular the Gram matrix comes:
degeneracy at any window would force a Brownian path to be affine on an interval of positive length, which the
quadratic variation forbids.

The argument also locates the role of the trimming. The restriction $r_2 - r_1 \geq r_0$ enters
twice, once in the uniform invertibility of $G$ and once through the compactness of $\mathcal{R}(r_0)$, and
Equation~\eqref{eq:uniformgram} fails at $r_0 = 0$, where windows of vanishing length drive $\lambda_{\min}(M)$ to
zero. Trimming is therefore part of the definition of the statistics, as Section~\ref{sec:simdesign} treats it.

Before stating the three theorems, we record the convergence on which all of them rest. Each statistic is an
infimum of $\textit{EG}(r_1,r_2)$ over a subset of $\mathcal{R}(r_0)$, so what is needed is not convergence of
$\textit{EG}(r_1,r_2)$ for each window separately, but convergence of the entire process indexed by
$(r_1,r_2)$. We state it once here and then apply it three times.

\begin{lemma}[Uniform convergence of the window statistic]\label{lem:uniform}
    Let $x_t$, $y_t$, and $\varepsilon_t$ satisfy Definition~\ref{def:model}, let Assumption~\ref{ass:fclt} hold,
    and let $H_0$ in Equation~\eqref{eq:hypotheses} hold. Then
    \begin{equation}\label{eq:uniformconv}
        \left\{ \textit{EG}(r_1,r_2) \right\}_{(r_1,r_2)\in\mathcal{R}(r_0)} \Rightarrow \left\{ \mathcal{Q}(r_1,r_2) \right\}_{(r_1,r_2)\in\mathcal{R}(r_0)}
    \end{equation}
    as random elements of $\ell^{\infty}(\mathcal{R}(r_0))$, the space of bounded functions on
    $\mathcal{R}(r_0)$ under the uniform metric, with $\mathcal{Q}$ as in Equation~\eqref{eq:Qfunctional}.
    Consequently, for any $\mathcal{W}\subseteq\mathcal{R}(r_0)$,
    \begin{equation*}
        \inf_{(r_1,r_2)\in\mathcal{W}}\textit{EG}(r_1,r_2) \xrightarrow{d} \inf_{(r_1,r_2)\in\mathcal{W}}\mathcal{Q}(r_1,r_2).
    \end{equation*}
\end{lemma}

\begin{proof}[Proof of Lemma~\ref{lem:uniform}]
    The argument has four steps. Only the third uses anything specific to the present setting.

    \emph{Step 1. One functional limit, not one per window.} Let
    $\xi_T(s) := T^{-1/2}\sum_{t\leq\lfloor Ts\rfloor}(\varepsilon_{x,t}^\top,\varepsilon_{y,t})^\top$ denote the
    normalized partial-sum process. Assumption~\ref{ass:fclt}(i) gives
    $\xi_T \Rightarrow (B_x^\top,B_y)^\top$ in $D[0,1]^{N+1}$. Every quantity entering
    $\textit{EG}(r_1,r_2)$ (the window Gram matrix, the window cross-moment, the estimated coefficient, the
    recursive residuals, and the numerator and denominator of the $t$-statistic) is a function of $\xi_T$
    restricted to $[r_1,r_2]$, together with sample second moments of the increments. There is therefore a
    single weak convergence to establish, and the window index enters only through the functional applied to it.

    \emph{Step 2. Window functionals converge uniformly in the endpoints.} Writing $z_t = (d_t^\top,x_t^\top)^\top$ as in
    Equation~\eqref{eq:recursive_beta}, the normalized window moments satisfy
    \begin{equation*}
        \sup_{(r_1,r_2)\in\mathcal{R}(r_0)} \left\lVert D_T^{-1}\!\!\sum_{t=\lfloor Tr_1\rfloor+1}^{\lfloor Tr_2\rfloor}\!\! z_t z_t^\top D_T^{-1} - \int_{r_1}^{r_2} \zeta(u)\zeta(u)^\top \mathrm{d}u \right\rVert \xrightarrow{p} 0 ,
    \end{equation*}
    with $D_T := \operatorname{diag}(T^{1/2}I_q,\, T I_N)$ and $\zeta = (d^\top, B_x^\top)^\top$, and similarly for the
    cross-moment and for $T^{-1}\sum_t (\Delta\hat{\varepsilon}_{t,r_1,r_2})^2$. For the terms that are Riemann sums, this follows from the
    previous step by the continuous-mapping theorem, since $(r_1,r_2)\mapsto\int_{r_1}^{r_2}$ is continuous on
    $\mathcal{R}(r_0)$; for the terms of stochastic-integral type it follows from the convergence of the partial-sum
    stochastic integral as a process in its upper limit \citep{Hansen1992ConvergenceProcesses}, which is what
    supplies the uniformity that a pointwise argument cannot; uniformity in the lower endpoint follows from the same
    statement, since $\int_{r_1}^{r_2} = \int_0^{r_2} - \int_0^{r_1}$.
    \citet{Hansen1992TestsProcesses} uses these same uniform-in-endpoint results for recursive regressions
    with $I(1)$ regressors, which is the closest antecedent for the objects appearing here.

    \emph{Step 3. The map is continuous, uniformly over the trimmed set.} Combining the two previous steps,
    $\textit{EG}(r_1,r_2) = \Phi(\xi_T)(r_1,r_2) + o_p(1)$ uniformly on $\mathcal{R}(r_0)$, where $\Phi$ sends a
    continuous path to the right-hand side of Equation~\eqref{eq:Qfunctional}; the lag augmentation contributes the
    usual correction, which under Assumption~\ref{ass:fclt}(ii) annihilates the one-sided long-run covariance and
    replaces the innovation variance by $\omega_{\hat{\varepsilon},r_1,r_2}^2$, and does so at a rate that does not
    depend on the window because every window has length at least $r_0$. The map $\Phi$ is continuous at continuous
    paths in the uniform topology provided three denominators are bounded away from zero uniformly on
    $\mathcal{R}(r_0)$:
    \begin{enumerate}[label=(\alph*),leftmargin=2.2em,topsep=3pt,itemsep=1pt,parsep=0pt]
        \item the deterministic Gram matrix $\int_{r_1}^{r_2} d\,d^\top$, whose smallest eigenvalue is
              bounded below by a constant depending only on $r_0$ and on the case in
              Equation~\eqref{eq:detcases}, since $r_2-r_1\geq r_0$;
        \item the window Gram matrix, by Lemma~\ref{lem:gram}; and
        \item $\omega_{\hat{\varepsilon},r_1,r_2}^2$, bounded below by $\lambda_{\min}(\Omega)>0$ at every
              window by Equation~\eqref{eq:omegabound}.
    \end{enumerate}
    This is where the trimming does its
    work: at $r_0 = 0$ the first of the three bounds fails, the window regressions degenerate, and no such
    statement is available. Continuity of $\Phi$ and Step~1 give Equation~\eqref{eq:uniformconv} by the
    continuous-mapping theorem in $\ell^{\infty}(\mathcal{R}(r_0))$ \citep{VanDerVaart1996WeakProcesses}.

    \emph{Step 4. The infimum.} For any $\mathcal{W}\subseteq\mathcal{R}(r_0)$ the map $f\mapsto\inf_{\mathcal{W}} f$ is
    $1$-Lipschitz with respect to the uniform norm, hence continuous, and a second application of the
    continuous-mapping theorem gives the final claim.
\end{proof}

Lemmas~\ref{lem:gram} and~\ref{lem:uniform} together derive the functional convergence from
Assumption~\ref{ass:fclt}: a functional central limit theorem with a non-singular long-run covariance matrix, and a
deterministic lag rule. The second step of the proof invokes, rather than re-derives, the classical
uniform-in-endpoint convergence results for $I(1)$ regression functionals cited there.

Each ingredient has an antecedent. \citet{Phillips2011BehaviorValues, Phillips2015TestingSP500} take suprema of
recursive unit-root statistics over the same doubly-indexed, trimmed window set, although their statistic is a
functional of the observed series with no parameter estimated inside the window, so their limit map is continuous
and no analogue of Lemma~\ref{lem:gram} arises. \citet{Gregory1996Residual-basedShifts} take an infimum of a
residual-based statistic in which the cointegrating vector is estimated, which is the feature that generates
the window Gram matrix, but over a single index. \citet{Andrews1993TestsPoint} supplies the architecture of a
trimmed index set followed by a continuous-mapping step for an extremum over that set, and
\citet{Banerjee1992RecursiveEvidence} the recursive minimum Dickey-Fuller statistic. The statistics studied here
combine the two features: a doubly-indexed infimum of a statistic with a window-specific estimated coefficient.

Lemma~\ref{lem:uniform} is stated under $H_0$, and the behavior under $H_1$ needs two further
consequences of Assumption~\ref{ass:fclt}. The first says what happens on a window contained in the
stationary segment, and is what drives every consistency statement in the paper. The second says what
happens on windows kept away from that segment.

\begin{lemma}[Divergence on a clean window]\label{lem:clean}
    Let Definition~\ref{def:model} and Assumption~\ref{ass:fclt} hold, let $\varepsilon_t$ satisfy
    Definition~\ref{def:coint} with stationary segment $\mathcal{S}$ as in Equation~\eqref{eq:statseg}, and
    let $(r_1,r_2)\in\mathcal{R}(r_0)$ be \emph{clean}, that is $[r_1,r_2]\subseteq\mathcal{S}$. Then there
    is a constant $c_{r_1,r_2}>0$ such that
    \begin{equation}\label{eq:adf_divergence}
        \textit{EG}(r_1,r_2) = -c_{r_1,r_2}\sqrt{T}\,\{1+o_p(1)\} \xrightarrow{p} -\infty .
    \end{equation}
\end{lemma}

\begin{proof}[Proof of Lemma~\ref{lem:clean}]
    On a clean window $\varepsilon_t$ is $I(0)$ throughout, so the window regression is a genuine
    cointegrating regression and, under the standard conditions for residual-based cointegration estimation,
    its OLS estimator is super-consistent,
    \begin{equation}\label{eq:superconsistency}
        \hat{\beta}_{r_1,r_2} - \beta = O_p(T^{-1}), \qquad \hat{\theta}_{r_1,r_2} - \theta = O_p(T^{-1/2}),
        \qquad [r_1,r_2]\subseteq\mathcal{S}.
    \end{equation}
    Super-consistency is specific to clean windows: if $[r_1,r_2]$ also covers part of an $I(1)$ segment, the
    regression is spurious over a non-vanishing fraction of the window, $\hat{\beta}_{r_1,r_2}$ converges at a
    slower rate, and the argument that follows does not apply.

    Since $\max_{t\leq T}\lVert x_t\rVert = O_p(\sqrt{T})$ for an $I(1)$ process and $d_t = O(1)$ under
    Equation~\eqref{eq:detcases}, Equation~\eqref{eq:superconsistency} bounds the estimation error in the
    window residuals uniformly in $t$,
    \begin{equation}\label{eq:residual_error}
        \max_{t\leq T}\bigl\lvert \hat{\varepsilon}_{t,r_1,r_2} - \varepsilon_t \bigr\rvert
          = \max_{t\leq T}\bigl\lvert (\hat{\theta}_{r_1,r_2}-\theta)^\top d_t
            + (\hat{\beta}_{r_1,r_2}-\beta)^\top x_t \bigr\rvert = O_p(T^{-1/2}),
    \end{equation}
    so estimating the coefficient vector does not alter the leading asymptotic behavior of the residual
    process. The level coefficient in the ADF regression of Equation~\eqref{eq:residmodel}, run on
    $\hat{\varepsilon}_{t,r_1,r_2}$ over the same window, therefore behaves as it would on $\varepsilon_t$
    itself,
    \begin{equation*}
        \hat{\gamma}_{r_1,r_2} \xrightarrow{p} \gamma^{*}_{r_1,r_2} = \phi^{*}_{r_1,r_2} - 1 < 0,
        \qquad \lvert\phi^{*}_{r_1,r_2}\rvert < 1,
    \end{equation*}
    with a standard error of exact order $T^{-1/2}$, that is
    $\operatorname{se}(\hat{\gamma}_{r_1,r_2}) = \varsigma_{r_1,r_2}\,T^{-1/2}\{1+o_p(1)\}$ for a constant
    $\varsigma_{r_1,r_2} > 0$. The window contains $\lfloor T(r_2-r_1)\rfloor$ observations with
    $r_2-r_1\geq r_0$, so its length enters $\varsigma_{r_1,r_2}$ but not the rate. Equation~\eqref{eq:adf_divergence}
    follows with $c_{r_1,r_2} = -\gamma^{*}_{r_1,r_2}/\varsigma_{r_1,r_2} > 0$. An upper bound on the standard
    error would not suffice here: it is the exact rate that turns a fixed negative probability limit for
    $\hat{\gamma}_{r_1,r_2}$ into divergence of the ratio.
\end{proof}

Lemma~\ref{lem:clean} requires containment of the window within the stationary segment, not the
presence of a stationary sub-segment inside the window. A window covering both regimes carries a unit
root over a non-vanishing fraction, and the next lemma shows that such windows contribute nothing to the
divergence, provided they are kept away from $\mathcal{S}$ uniformly.

\begin{lemma}[Tightness away from the stationary segment]\label{lem:tight}
    Let Definition~\ref{def:model} and Assumption~\ref{ass:fclt} hold and let $\varepsilon_t$ satisfy
    Definition~\ref{def:coint} with stationary segment $\mathcal{S}$. Let $\mathcal{W}\subseteq\mathcal{R}(r_0)$
    be such that, for some $\underline{\ell}>0$,
    \begin{equation}\label{eq:elowbound}
        \bigl\lvert [r_1,r_2]\setminus\mathcal{S} \bigr\rvert \;\geq\; \underline{\ell}
        \qquad\text{for every } (r_1,r_2)\in\mathcal{W}.
    \end{equation}
    Then $\sup_{(r_1,r_2)\in\mathcal{W}} \lvert\textit{EG}(r_1,r_2)\rvert = O_p(1)$.
\end{lemma}

\begin{proof}[Proof of Lemma~\ref{lem:tight}]
    Only tightness is claimed, not a limit, and this is what makes the argument short: Definition~\ref{def:coint}
    leaves the limiting residual process discontinuous at $\tau_0$, so no continuous-mapping statement of the kind
    used in Lemma~\ref{lem:uniform} is available, and none is needed.

    Definition~\ref{def:coint} changes the residual process, not the regressors, so $B_x$ is as under $H_0$ and
    Lemma~\ref{lem:gram} applies verbatim. The normalized partial-sum process $\xi_T$ of the proof of
    Lemma~\ref{lem:uniform} remains tight, since
    $T^{-1/2}\sum_{t\leq\lfloor Ts\rfloor}\Delta\varepsilon_t = T^{-1/2}(\varepsilon_{\lfloor Ts\rfloor}-\varepsilon_0)$
    is $O_p(1)$ uniformly in $s$ under either direction of Definition~\ref{def:coint}, and the deterministic
    increments contribute $O(T^{-1})$. Step~2 of that proof writes every ingredient of $\textit{EG}(r_1,r_2)$ as a
    functional of $\xi_T$ restricted to $[r_1,r_2]$, bounded in the uniform norm by a constant multiple of
    $\sup_{s}\lVert\xi_T(s)\rVert$; the numerator of the statistic is therefore $O_p(1)$ uniformly on
    $\mathcal{R}(r_0)$.

    It remains to bound the denominator away from zero uniformly on $\mathcal{W}$. Its three factors are
    \begin{enumerate}[label=(\alph*),leftmargin=2.2em,topsep=3pt,itemsep=1pt,parsep=0pt]
        \item the deterministic Gram matrix, bounded below by a constant depending only on $r_0$ and
              the case in Equation~\eqref{eq:detcases};
        \item the window Gram matrix, by Lemma~\ref{lem:gram}; and
        \item the estimated innovation variance of the ADF regression.
    \end{enumerate}
    For the last, Equation~\eqref{eq:elowbound} places a sub-interval of
    $[r_1,r_2]$ of length at least $\underline{\ell}$ outside $\mathcal{S}$, and over it the residual increments
    carry the unit-root innovations, whose variance conditional on the regressor increments is positive by the
    positive definiteness of $\Omega$ in Assumption~\ref{ass:fclt}(i); the window average is therefore bounded
    below by a positive multiple of $\underline{\ell}$, uniformly on $\mathcal{W}$. Hence
    $\sup_{\mathcal{W}}\lvert\textit{EG}\rvert$ is $O_p(1)$.

    Equation~\eqref{eq:elowbound} is what the argument needs, and it is strictly more than the requirement that
    $\mathcal{W}$ contain no clean window: a set of windows none of which is clean may still approach a clean
    one, and along such a sequence $\textit{EG}$ need not be bounded in probability.
    Proposition~\ref{prop:consistency} verifies Equation~\eqref{eq:elowbound} in each case where it is invoked.
\end{proof}

\subsection{Proof of the Main Theorem and Invariance of the Limit}\label{app:proofFIEG}

\begin{proof}[Proof of Theorem~\ref{thm:FIEG}]
    The proof combines the FCLT with standard results for unit-root and residual-based cointegration
    asymptotics \citep{Mann1943OnRelationships,Donsker1951AnTheorems,Dickey1979DistributionRoot,Phillips1987TowardsAutoregression,Phillips1988TestingRegression,Phillips1990AsymptoticCointegration,Phillips2015TestingDetectors}.
    Under $H_0$ the residual process is integrated of order one. Consider the recursive OLS estimator in
    Equation~\eqref{eq:recursive_beta}. Since $x_t$ and $y_t$ are integrated processes, the FCLT, together
    with the corresponding Riemann-sum convergence and the Continuous Mapping Theorem [CMT], yields
    \begin{equation}\label{eq:beta_limit}
        \hat{\beta}_r \Rightarrow b_r = \left[\int_0^r B_x^{\perp}(u) B_x^{\perp}(u)^\top \mathrm{d}u \right]^{-1} \int_0^r B_x^{\perp}(u) B_y^{\perp}(u) \mathrm{d}u, \quad r\in[r_0,1],
    \end{equation}
    and the recursive residuals in Equation~\eqref{eq:recursive_residual} satisfy the process convergence
    \begin{equation}\label{eq:recursive_residual_limit}
        T^{-1/2}\,\hat{\varepsilon}_{\lfloor Ts \rfloor, r} \Rightarrow B_{\hat{\varepsilon},r}(s), \qquad 0 \leq s \leq r,
    \end{equation}
    with $B_{\hat{\varepsilon},r}$ as defined in Equation~\eqref{eq:limit_objects}. The dependence of
    $B_{\hat{\varepsilon},r}(s)$ on the recursive endpoint $r$ reflects the fact that the coefficient used to
    construct the residual process is estimated using observations up to $\lfloor Tr \rfloor$. Applying the
    standard Dickey-Fuller argument to that process, with the lag-order condition following
    Equation~\eqref{eq:residmodel} annihilating the one-sided long-run covariance term and replacing the
    innovation variance by $\omega_{\hat{\varepsilon},r}^{2}$ in the denominator, gives
    $\textit{EG}(r) \xrightarrow{d} \mathcal{Q}(0,r)$ for each $r\in[r_0,1]$.

    Equations~\eqref{eq:beta_limit} and~\eqref{eq:recursive_residual_limit} are the case $r_1=0$ of the
    convergence established in the proof of Lemma~\ref{lem:uniform}, which upgrades the pointwise statement to
    the process indexed by the recursive endpoint. Since the $\textit{FIEG}$ search set
    $\{(0,r) : r \in [r_0,1]\}$ is a subset of $\mathcal{R}(r_0)$, the final claim of that lemma gives
    \begin{equation}\label{eq:fieg_limit_proof}
        \textit{FIEG}(r_0) = \inf_{r \in [r_0,1]}\textit{EG}(r) \xrightarrow{d} \inf_{r \in [r_0,1]} \mathcal{Q}(0,r).
    \end{equation}

    Under the alternative hypothesis $H_1$ in Equation~\eqref{eq:hypotheses}, Lemma~\ref{lem:clean} applies to
    any clean recursive window, that is to any $r\in[r_0,1]$ with $[0,r]\subseteq\mathcal{S}$, and gives
    $\textit{EG}(r)\xrightarrow{p}-\infty$; one such window is therefore enough to drive
    $\textit{FIEG}(r_0)$ to $-\infty$. Under a forward break, Equation~\eqref{eq:epsbehave},
    $\mathcal{S}=[0,\tau_0]$ and every $r\in[r_0,\tau_0]$ gives a clean window, so
    $\textit{FIEG}(r_0)\xrightarrow{p}-\infty$ whenever $\tau_0\geq r_0$. Under a reverse break,
    Equation~\eqref{eq:epsbehave_reverse}, the situation is different: $\mathcal{S}=[\tau_0,1]$ with
    $\tau_0>0$, and every recursive window $[0,r]$ contains the initial $I(1)$ segment $[0,\tau_0)$. No
    admissible window is then clean, Equation~\eqref{eq:elowbound} holds on the $\textit{FIEG}$ search set with
    $\underline{\ell}=\tau_0$, and Lemma~\ref{lem:tight} gives $\textit{FIEG}(r_0)=O_p(1)$, so the test has no
    power. This directional limitation is intrinsic to the forward-recursive design rather than a matter of
    sample size; see Proposition~\ref{prop:consistency}.
\end{proof}

\begin{lemma}[Invariance to nuisance parameters]\label{lem:pivotal}
    Let Assumption~\ref{ass:fclt} and the null hypothesis $H_0$ in Equation~\eqref{eq:hypotheses} hold, so that both blocks are exactly $I(1)$. Then the limit in Theorem~\ref{thm:FIEG} does not depend on the long-run covariance matrix $\Omega$. It depends only on $N$, on $r_0$, and on the deterministic case in Equation~\eqref{eq:detcases}.
\end{lemma}

\begin{proof}[Proof of Lemma~\ref{lem:pivotal}]
    The argument is a change of variables in two moves. The first partials $B_y$ on $B_x$, so that the
    nuisance covariance enters only through a term that cancels; the second rescales both processes to
    standard Brownian motions, so that the remaining scale factor cancels as well.

    Decompose $B_y = \Omega_{yx}\Omega_{xx}^{-1} B_x + B_{y\cdot x}$, where $B_{y\cdot x}$ is a scalar Brownian motion independent of $B_x$ with variance rate $\omega_{y\cdot x}^{2} = \omega_{yy} - \Omega_{yx}\Omega_{xx}^{-1}\Omega_{xy}$. The projection in Equation~\eqref{eq:projection} is linear, so the decomposition is preserved by it, $B_y^{\perp} = \Omega_{yx}\Omega_{xx}^{-1} B_x^{\perp} + B_{y\cdot x}^{\perp}$. Substituting into Equation~\eqref{eq:beta_limit} gives
    \begin{equation*}
        b_r = \Omega_{xx}^{-1}\Omega_{xy} + \left[\int_0^r B_x^{\perp} B_x^{\perp\top} \mathrm{d}u\right]^{-1}\int_0^r B_x^{\perp} B_{y\cdot x}^{\perp}\, \mathrm{d}u,
    \end{equation*}
    so that the first term cancels in Equation~\eqref{eq:limit_objects} and
    \begin{equation*}
        B_{\hat{\varepsilon},r}(s) = B_{y\cdot x}^{\perp}(s) - B_x^{\perp}(s)^\top c_r, \qquad c_r := \left[\int_0^r B_x^{\perp} B_x^{\perp\top} \mathrm{d}u\right]^{-1}\int_0^r B_x^{\perp} B_{y\cdot x}^{\perp}\, \mathrm{d}u.
    \end{equation*}

    Now write $B_x = \Omega_{xx}^{1/2} V$ and $B_{y\cdot x} = \omega_{y\cdot x} U$, where $V$ is a standard $N$-dimensional and $U$ a standard scalar Brownian motion, mutually independent. Since the projection commutes with multiplication by a constant matrix, $B_x^{\perp} = \Omega_{xx}^{1/2} V^{\perp}$ and $B_{y\cdot x}^{\perp} = \omega_{y\cdot x} U^{\perp}$. Then $B_x^{\perp}(s)^\top c_r = \omega_{y\cdot x}\, V^{\perp}(s)^\top \tilde{c}_r$ with $\tilde{c}_r = \left[\int_0^r V^{\perp} V^{\perp\top} \mathrm{d}u\right]^{-1}\int_0^r V^{\perp} U^{\perp} \,\mathrm{d}u$, and hence
    \begin{equation*}
        B_{\hat{\varepsilon},r}(s) = \omega_{y\cdot x}\left\{ U^{\perp}(s) - V^{\perp}(s)^\top \tilde{c}_r \right\}, \qquad \omega_{\hat{\varepsilon},r}^{2} = \omega_{y\cdot x}^{2}\left(1 + \tilde{c}_r^\top \tilde{c}_r\right),
    \end{equation*}
    where the expression for $\omega_{\hat{\varepsilon},r}^{2}$ is unaffected by the projection, since subtracting a process of finite variation leaves quadratic variation unchanged. The factor $\omega_{y\cdot x}$ therefore cancels in $W_{\hat{\varepsilon},r} = B_{\hat{\varepsilon},r}/\omega_{\hat{\varepsilon},r}$, giving
    \begin{equation}\label{eq:pivotal_form}
        W_{\hat{\varepsilon},r}(s) = \frac{U^{\perp}(s) - V^{\perp}(s)^\top \tilde{c}_r}{\sqrt{1 + \tilde{c}_r^\top \tilde{c}_r}},
    \end{equation}
    which is a functional of standard Brownian motions and the deterministic function $d$ alone, and therefore free of $\Omega$. Since $\mathcal{Q}(r_1,r_2)$ in Equation~\eqref{eq:Qfunctional} depends on the underlying processes only through $W_{\hat{\varepsilon},r_1,r_2}$, the same is true of the limit itself. The argument applies verbatim on the window $[r_1,r_2]$, so the limits in Theorem~\ref{thm:BIEG} and Theorem~\ref{thm:GIEG} are likewise free of $\Omega$.
\end{proof}

\subsection{Consistency and Lag Selection}\label{app:consistency}

\begin{proposition}[Directional consistency]\label{prop:consistency}
    Let Definition~\ref{def:model} and Assumption~\ref{ass:fclt} hold, and let $\varepsilon_t$ satisfy
    Definition~\ref{def:coint} for some $\tau_0\in(0,1)$, with stationary segment $\mathcal{S}$ as in
    Equation~\eqref{eq:statseg}. Consider a statistic of the form
    $\inf_{(r_1,r_2)\in\mathcal{W}}\textit{EG}(r_1,r_2)$ for a search set $\mathcal{W}\subseteq\mathcal{R}(r_0)$.
    Then the statistic diverges to $-\infty$ in probability if $\mathcal{W}$ contains a clean window, and is
    $O_p(1)$ if $\mathcal{W}$ satisfies Equation~\eqref{eq:elowbound} for some $\underline{\ell}>0$.
    Consequently,
    \begin{enumerate}[label=(\roman*)]
        \item $\textit{FIEG}(r_0) \xrightarrow{p} -\infty$ under a forward break with $\tau_0 \geq r_0$, and $\textit{FIEG}(r_0) = O_p(1)$ under a reverse break, for every $\tau_0\in(0,1)$;
        \item $\textit{BIEG}(r_0,1) \xrightarrow{p} -\infty$ under a reverse break with $1-\tau_0 \geq r_0$, and $\textit{BIEG}(r_0,1) = O_p(1)$ under a forward break, for every $\tau_0\in(0,1)$;
        \item $\textit{GIEG}(r_0) \xrightarrow{p} -\infty$ under either break direction, provided only that $\lvert\mathcal{S}\rvert \geq r_0$.
    \end{enumerate}
\end{proposition}

\begin{proof}[Proof of Proposition~\ref{prop:consistency}]
    The two general claims are Lemmas~\ref{lem:clean} and~\ref{lem:tight}: a clean window in $\mathcal{W}$
    contributes an $\textit{EG}(r_1,r_2)$ of order $-\sqrt{T}$, so the infimum diverges; and under
    Equation~\eqref{eq:elowbound} the whole family is uniformly bounded in probability, so the infimum is
    $O_p(1)$. It remains to check which of the two applies to each search set.

    For (i), the $\textit{FIEG}$ search set is $\{(0,r): r\in[r_0,1]\}$. Under a forward break
    $\mathcal{S}=[0,\tau_0]$ and $(0,r)$ is clean for every $r\in[r_0,\tau_0]$, a non-empty range when
    $\tau_0\geq r_0$. Under a reverse break $\mathcal{S}=[\tau_0,1]$ with $\tau_0>0$, so every $[0,r]$
    contains $[0,\tau_0)$ and Equation~\eqref{eq:elowbound} holds with $\underline{\ell}=\tau_0$. Part (ii)
    is the mirror image, with search set $\{(r_1,1): r_1\in[0,1-r_0]\}$: under a reverse break with
    $1-\tau_0\geq r_0$ the window $(r_1,1)$ is clean for every $r_1\in[\tau_0,1-r_0]$, and under a forward
    break every $[r_1,1]$ contains $[\tau_0,1]$, so Equation~\eqref{eq:elowbound} holds with
    $\underline{\ell}=1-\tau_0$. For (iii), $\mathcal{R}(r_0)$ contains every window of length $r_0$, and
    $\mathcal{S}$ contains such a window whenever $\lvert\mathcal{S}\rvert\geq r_0$.

    The uniform separation in Equation~\eqref{eq:elowbound} is available in (i) and (ii) because a restricted
    search set pins one endpoint, which keeps every window a fixed distance from $\mathcal{S}$. It is not
    available for $\mathcal{R}(r_0)$ itself, and no statement is needed there: by (iii) the $\textit{GIEG}$
    search set always contains a clean window under $H_1$.
\end{proof}

Finally, we return to the lag order, which Assumption~\ref{ass:fclt}(ii) fixes by a deterministic
rule while the implementation selects it from the data.

\begin{remark}[Data-dependent lag selection]\label{rem:bic}
    Let $\hat{p}(r_1,r_2)$ denote the order selected by BIC in the window $[r_1,r_2]$, let
    $\mathcal{R}_T(r_0) \subset \mathcal{R}(r_0)$ denote the finite collection of windows
    actually scanned at sample size $T$, and define the event
    \begin{equation*}
        A_T := \left\{ \hat{p}(r_1,r_2) = \bar{p} \ \text{ for all } (r_1,r_2) \in \mathcal{R}_T(r_0) \right\}.
    \end{equation*}
    On $A_T$ the statistic process computed with selected lags coincides, path by path,
    with the one computed under the deterministic rule, and hence so do their infima.
    Consequently, if $\mathbb{P}(A_T) \to 1$, the two are asymptotically equivalent and
    Theorems~\ref{thm:FIEG}, \ref{thm:BIEG}, and \ref{thm:GIEG} apply unchanged. This
    argument requires no continuity of the selected order in $(r_1,r_2)$, only that the
    selection be eventually constant across windows.

    When $\varepsilon_t$ admits a finite autoregressive representation, BIC is order
    consistent in each window, since every admissible window contains at least
    $\lfloor r_0 T \rfloor$ observations and consistency of the criterion is preserved in
    autoregressions with unit roots \citep{Paulsen1984OrderRoots}. Uniformity over the
    $O(T^{2})$ windows scanned by $\textit{GIEG}$ requires in addition that the per-window
    misselection probability be $o(T^{-2})$, which we assume rather than establish.
    Pointwise, the limiting distribution of the ADF $t$-statistic is known to be invariant
    to data-dependent lag selection, provided the selected order satisfies the rate
    condition in probability \citep{Ng1995UnitLag, Chang2002ONROOTS}. The
    argument above delivers the uniform version by a different route, and a cheaper one: on
    $A_T$ the two processes are identical, so no invariance of the limit under a
    window-varying order is needed at all.
\end{remark}

\section{Theorem and Proof for BIEG and GIEG}\label{app:BIEGandGIEG}

This appendix presents the formal theorems for the BIEG and GIEG statistics. Both are computed from the
window estimator and window residuals of Equation~\eqref{eq:beta_window} and differ from $\textit{FIEG}$, and
from each other, only in the set over which the infimum is taken. Since Lemma~\ref{lem:uniform} delivers the
convergence of the whole process $\{\textit{EG}(r_1,r_2)\}$ on $\mathcal{R}(r_0)$ and its final claim covers
every subset of it, each theorem is an application of that lemma to the relevant search set, and the
substance of each proof lies in the behavior under the alternative.

\begin{thm}[Backward-Infimum Engle-Granger]\label{thm:BIEG}
    Let $x_t, y_t,$ and $\varepsilon_t$ satisfy Definition~\ref{def:model}, and let Assumption~\ref{ass:fclt} hold. Then, under the null hypothesis $H_0$ in Equation~\eqref{eq:hypotheses}, for each fixed $r_2 \in [r_0, 1]$,
    \begin{align}\label{eq:asympnullBIEG}
        \textit{BIEG}(r_0, r_2) \xrightarrow{d} \inf_{r_1 \in [0, r_2 - r_0]} \mathcal{Q}(r_1, r_2). 
    \end{align}    
\end{thm}

\begin{proof}[Proof of Theorem~\ref{thm:BIEG}]
    Under $H_0$ the residual process is integrated of order one on every window. As in
    Equation~\eqref{eq:beta_limit}, the FCLT together with the Riemann-sum convergence and the CMT gives
    $\hat{\beta}_{r_1,r_2} \Rightarrow b_{r_1,r_2}$ and
    $T^{-1/2}\hat{\varepsilon}_{\lfloor Ts\rfloor,r_1,r_2} \Rightarrow B_{\hat{\varepsilon},r_1,r_2}(s)$ for
    $s\in[r_1,r_2]$, with the limits of Equation~\eqref{eq:limit_objects}; the standard Dickey-Fuller argument,
    with the lag-order condition following Equation~\eqref{eq:residmodel} annihilating the one-sided long-run
    covariance term and replacing the innovation variance by $\omega_{\hat{\varepsilon},r_1,r_2}^{2}$, gives
    $\textit{EG}(r_1,r_2) \xrightarrow{d} \mathcal{Q}(r_1,r_2)$. Lemma~\ref{lem:uniform} upgrades this to the
    process indexed by $(r_1,r_2)$ on the whole of $\mathcal{R}(r_0)$. For fixed $r_2\in[r_0,1]$ the
    $\textit{BIEG}$ search set $\{(r_1,r_2) : r_1 \in [0, r_2-r_0]\}$ is a subset of $\mathcal{R}(r_0)$, so the
    final claim of that lemma yields Equation~\eqref{eq:asympnullBIEG}.

    Under $H_1$, Lemma~\ref{lem:clean} gives $\textit{EG}(r_1,r_2)\xrightarrow{p}-\infty$ on any clean window,
    so for a given $r_2$ the BIEG statistic diverges whenever some $r_1\in[0,r_2-r_0]$ makes
    $[r_1,r_2]\subseteq\mathcal{S}$. The endpoint $r_2$ therefore determines against which break direction the
    test has power. In the specification used throughout this paper, $r_2=1$, so every window is of the form
    $[r_1,1]$. Under a reverse break, Equation~\eqref{eq:epsbehave_reverse}, $\mathcal{S}=[\tau_0,1]$ and any
    $r_1\in[\tau_0,1-r_0]$ gives a clean window, so $\textit{BIEG}(r_0,1)\xrightarrow{p}-\infty$ whenever
    $1-\tau_0\geq r_0$. Under a forward break, Equation~\eqref{eq:epsbehave}, $\mathcal{S}=[0,\tau_0]$ with
    $\tau_0<1$ and every window $[r_1,1]$ contains the terminal $I(1)$ segment $[\tau_0,1]$; no admissible
    window is clean, Lemma~\ref{lem:tight} applies with $\underline{\ell}=1-\tau_0$, and
    $\textit{BIEG}(r_0,1)=O_p(1)$. For a general fixed $r_2<1$ the roles are partially reversed, since a clean
    window exists under a forward break whenever $r_0 \leq r_2 \leq\tau_0$. See
    Proposition~\ref{prop:consistency}.
\end{proof}

\begin{thm}[Generalized-Infimum Engle-Granger]\label{thm:GIEG}
    Let $x_t$, $y_t$, and $\varepsilon_t$ satisfy Definition~\ref{def:model} and let Assumption~\ref{ass:fclt} hold. Then, under the null hypothesis $H_0$ in Equation~\eqref{eq:hypotheses}, it holds that
    \begin{align}\label{eq:asympnullGIEG}
        \textit{GIEG}(r_0) \xrightarrow{d} \inf_{(r_1,r_2) \in \mathcal{R}(r_0)} \mathcal{Q}(r_1, r_2). 
    \end{align}    
\end{thm}

\begin{proof}[Proof of Theorem~\ref{thm:GIEG}]
    The null argument is the one just given for Theorem~\ref{thm:BIEG}, with no restriction on either endpoint:
    the $\textit{GIEG}$ search set is the whole of $\mathcal{R}(r_0)$, so Lemma~\ref{lem:uniform} applies to it
    directly and gives Equation~\eqref{eq:asympnullGIEG}.

    Under $H_1$, that same absence of restriction is what distinguishes $\textit{GIEG}$ from the other two
    statistics. Its search set contains a clean window whenever the stationary segment is long enough to hold
    one, that is whenever $\lvert\mathcal{S}\rvert\geq r_0$: under Equation~\eqref{eq:epsbehave_reverse} the
    window $(\tau_0,\,\tau_0+r_0)$ is clean whenever $1-\tau_0\geq r_0$, exactly as $(0,\tau_0)$ is clean under
    Equation~\eqref{eq:epsbehave} whenever $\tau_0\geq r_0$. Lemma~\ref{lem:clean} then gives
    $\textit{GIEG}(r_0)\xrightarrow{p}-\infty$ in both cases, so the test is consistent against both break
    directions; see Proposition~\ref{prop:consistency}.
\end{proof}

\end{document}